\documentclass[twocolumn, 10pt]{extarticle}
\pdfoutput=1	% needed for the arXiv to properly compile with pdflatex

\usepackage[superscript]{cite}
\usepackage{enumitem}

\usepackage{titlesec}
\titleformat*{\section}{\normalsize\bfseries}
\titleformat*{\subsection}{\normalsize\bfseries}
\usepackage{fullpage}
\usepackage{amsfonts, amssymb, amsmath, amsthm}
\usepackage{latexsym}
\usepackage[tracking=smallcaps]{microtype}	% for an arXiv submission, set \pdfoutput=1 near the top so it uses pdflatex
\usepackage{url}
\usepackage{color}
\definecolor{DarkGray}{rgb}{0.1,0.1,0.5}
\usepackage[colorlinks=true,breaklinks, linkcolor= DarkGray,citecolor= DarkGray,urlcolor= DarkGray]{hyperref}	% linkcolor=blue,citecolor=blue,urlcolor=blue

\usepackage{graphicx}
\usepackage[tight, TABBOTCAP]{subfigure}
\newcommand{\bra}[1]{{\langle#1|}}
\newcommand{\ket}[1]{{|#1\rangle}}

\newcommand{\ketbra}[2]{{\ket{#1}\!\bra{#2}}}

\newcommand{\abs}[1]{{\lvert #1\rvert}}	% since the delimiters do not scale, it might be a good idea to add a dummy {} at the end, so \abs{big expression}^2 has the superscript at a low height

\newcommand{\trnorm}[1]{{\| #1 \|_{\mathrm{tr}}}}
\newcommand{\smatrx}[1]{\ensuremath{\left(\begin{smallmatrix}#1\end{smallmatrix}\right)}}

\DeclareMathOperator{\Tr}{\operatorname{Tr}}

\def\A {{\mathcal A}}
\def\B {{\mathcal B}}
\def\C {{\bf C}}

\def\E {{\mathcal E}}
\def\F {{\mathcal F}}
\def\G {{\mathcal G}}
\def\H {{\mathcal H}}
\let\Lstroke\L	\def\L {{\mathcal L}}		

\def\cP {{\mathcal P}}

\def\S {{\mathcal S}}

\newcommand{\NP}{\ensuremath{\mathsf{NP}}}%{{\mathcal{NP}}}
\newcommand{\QMIP}{\ensuremath{\mathsf{QMIP}}}%{{\mathcal{NP}}}
\newcommand{\MIP}{\ensuremath{\mathsf{MIP}}}%{{\mathcal{NP}}}

\DeclareMathOperator{\poly}{\operatorname{poly}}

\newcommand{\identity}{\ensuremath{\boldsymbol{1}}} %\mathbb{I}

\def\phasegate {G}		%% careful changing this; we write ``a $\phasegate$ gate", and it is in \figref{f:combinedgadgetcircuit}

\newcommand{\kappaEPR}{\kappa_*}

\def\Aad {\A_{\text{ad}}}				% adaptive Alice-Eve super-operator for computation by teleportation
\def\Aadhat {\hat{\A}_{\text{ad}}}		% ideal adaptive Alice-Eve super-operator
\def\Bad {\B_{\text{ad}}}				% adaptive Bob-Eve super-operator
\def\Badhat {\hat{\B}_{\text{ad}}}		% ideal adaptive Bob-Eve super-operator

\def\device{D}	% used for an indeterminate prover: for $\device \in \{A, B\}$, $\H_\device$, etc.

\def\h #1{h_{#1}}
\def\hA #1{h_{#1}^{\smash{A}}}
\def\hB #1{h_{#1}^{\smash{B}}}
\def\hX #1{h_{#1}^{\smash{\device}}}
\def\hdec #1#2{#1{h}_{#2}}
\def\hAdec #1#2{#1{h}_{#2}^{\smash{A}}}
\def\hBdec #1#2{#1{h}_{#2}^{\smash{B}}}
\def\hXdec #1#2{#1{h}_{#2}^{\smash{\device}}}

\def\RAjah #1#2#3{R^A_{#2}({#3})}

\def\RBjah #1#2#3{R^B_{#2}({#3})}
\def\RAa #1{R^A_{#1}}
\def\RBa #1{R^B_{#1}}
\def\RXa #1{R^{\device}_{#1}}

\def\RBdeca #1#2{#1{R}^B_{#2}}

\def\RXjah #1#2#3{R^{\device}_{#2}({#3})}
\def\RXdecjah #1#2#3#4{#1{R}^{\device}_{#3}({#4})}
\def\PAjh #1#2{P^A_{#1}({#2})}		% with P_{j,k} notation, cannot hide the j argument
\def\PBjh #1#2{P^B_{#1}({#2})}

\def\PXjh #1#2{P^{\device}_{#1}({#2})}
\def\PABjh #1#2{P^{AB}_{#1}({#2})}

\def\EAj #1{\E^A_{#1}}
\def\EAdecj #1#2{#1{\E}^A_{#2}}
\def\EAjh #1#2{\E^{A \vert \smash{#2}}_{#1}}
\def\EBj #1{\E^B_{#1}}
\def\EBdecj #1#2{#1{\E}^B_{#2}}
\def\EXj #1{\E^{\device}_{#1}}
\def\EXdecj #1#2{#1{\E}^{\device}_{#2}}
\def\EBjh #1#2{\E^{B \vert \smash{#2}}_{#1}}
\def\EXjh #1#2{\E^{{\device} \vert \smash{#2}}_{#1}}
\def\EABj #1{\E^{AB}_{#1}}
\def\EABdecj #1#2{#1{\E}^{AB}_{#2}}
\def\EABjh #1#2{\E^{AB \vert \smash{#2}}_{#1}}

\def\Bguessj #1{\G^{B}_{#1}}

\def\UsingleXjh #1#2{U^{\device}_{#1}({#2})}

\def\UmultiXj #1{M^{\device}_{#1}}

\def\UmultiXjh #1#2{M^{\device}_{#1}({#2})}

\def\UidealA {\mathcal{I}^A}
\def\UidealB {\mathcal{I}^B}
\def\UidealX {\mathcal{I}^{\device}}
\def\UidealXdec #1{#1{\mathcal{I}}^{\device}}

\def\AAunitary {V}	% rotates from $\ket{(a, x)_A}$ to $\ket{a', x')_A}$
\def\AAunitaryAj #1{\AAunitary^A_{#1}}

\def\psione {\psi}

\def\rhoone{\rho_1}
\def\rhodecone #1{#1{\rho}_1}
\def\rhoj #1{\rho_{#1}}
\def\rhodecj #1#2{#1{\rho}_{#2}}
\def\XA {\mathcal{X}^A}
\def\XB {\mathcal{X}^B}

\def\XX {\mathcal{X}^{\device}}
\def\XmultiX {\mathcal{Y}^{\device}}

\newcounter{sprows}

\newlength{\spheight}
\newlength{\spraise}

\newlength{\commentslength}

\newcommand{\rem}[1]{}

\newtheorem{theorem}{Theorem}[section]
\newtheorem{lemma}[theorem]{Lemma}

\newtheorem{definition}[theorem]{Definition}

\newfont{\subsubsecfnt}{ptmri8t at 11pt}
\renewcommand{\subparagraph}[1]{\smallskip{\subsubsecfnt #1.}}

\newcommand{\eqnref}[1]{\hyperref[#1]{{(\ref*{#1})}}}
\newcommand{\thmref}[1]{\hyperref[#1]{{Theorem~\ref*{#1}}}}
\newcommand{\lemref}[1]{\hyperref[#1]{{Lemma~\ref*{#1}}}}
\newcommand{\corref}[1]{\hyperref[#1]{{Corollary~\ref*{#1}}}}
\newcommand{\defref}[1]{\hyperref[#1]{{Definition~\ref*{#1}}}}
\newcommand{\secref}[1]{\hyperref[#1]{{Section~\ref*{#1}}}}
\newcommand{\figref}[1]{\hyperref[#1]{{Figure~\ref*{#1}}}}
\newcommand{\tabref}[1]{\hyperref[#1]{{Table~\ref*{#1}}}}
\newcommand{\remref}[1]{\hyperref[#1]{{Remark~\ref*{#1}}}}
\newcommand{\appref}[1]{\hyperref[#1]{{Appendix~\ref*{#1}}}}
\newcommand{\claimref}[1]{\hyperref[#1]{{Claim~\ref*{#1}}}}
\newcommand{\factref}[1]{\hyperref[#1]{{Fact~\ref*{#1}}}}
\newcommand{\propref}[1]{\hyperref[#1]{{Proposition~\ref*{#1}}}}
\newcommand{\exampleref}[1]{\hyperref[#1]{{Example~\ref*{#1}}}}
\newcommand{\conjref}[1]{\hyperref[#1]{{Conjecture~\ref*{#1}}}}

\allowdisplaybreaks[1]
\usepackage{abstract}
\let\oldthebibliography=\thebibliography
\let\endoldthebibliography=\endthebibliography
\renewenvironment{thebibliography}[1]{%
  \begin{oldthebibliography}{#1}%
    \setlength{\parskip}{0ex}%
    \setlength{\itemsep}{0ex}%
}%
{%
  \end{oldthebibliography}%
}

\usepackage[letterpaper]{geometry}		

\begin{document}
\def\compilefullpaper{}

\title{\Large Classical command of quantum systems via rigidity of CHSH games} 
\author{Ben W.~Reichardt \\ {\small University of Southern California} \and Falk Unger \\ {\small Knight Capital Group} \and Umesh Vazirani \\ {\small UC Berkeley}}
\date{}

\twocolumn[\begin{@twocolumnfalse}
\maketitle
\begin{abstract}
Can a classical system command a general adversarial quantum system to realize arbitrary quantum dynamics?  If so, then we could realize the dream of device-independent quantum cryptography: using untrusted quantum devices to establish a shared random key, with security based on the correctness of quantum mechanics.  It would also allow for testing whether a claimed quantum computer is truly quantum.  Here we report a technique by which a classical system can certify the joint, entangled state of a bipartite quantum system, as well as command the application of specific operators on each subsystem.  This is accomplished by showing a strong converse to Tsirelson's optimality result for the Clauser-Horne-Shimony-Holt (CHSH) game: the only way to win many games is if the bipartite state is close to the tensor product of EPR states, and the measurements are the optimal CHSH measurements on successive qubits.  This leads directly to a scheme for device-independent quantum key distribution.  Control over the state and operators can also be leveraged to create more elaborate protocols for reliably realizing general quantum circuits.  
\end{abstract}
\end{@twocolumnfalse}]

Do the laws of quantum mechanics place any limits on how well a classical experimentalist can characterize the state and dynamics of a large quantum system?  As a thought experiment, consider that we are presented with a quantum system, together with instructions on how to control its evolution from a claimed initial state.  Can we, as classical beings, possibly convince ourselves that the quantum system was indeed initialized as claimed, and that its state evolves as we instruct? 

More formally, model the quantum system as contained in a black box, and model our classical interactions with it as questions and answers across a digital interface, perhaps of buttons and light bulbs (\figref{f:blackbox}).  Using this limited interface, we wish to characterize the initial state of the system.  We also wish to certify that on command---by pressing a suitable sequence of buttons---the system applies a chosen local Hamiltonian, or equivalently a sequence of one- and two-qubit quantum gates, and outputs desired measurement results.  

\begin{figure}[b!]
\centering
\raisebox{-.25cm}{\includegraphics[scale=.5]{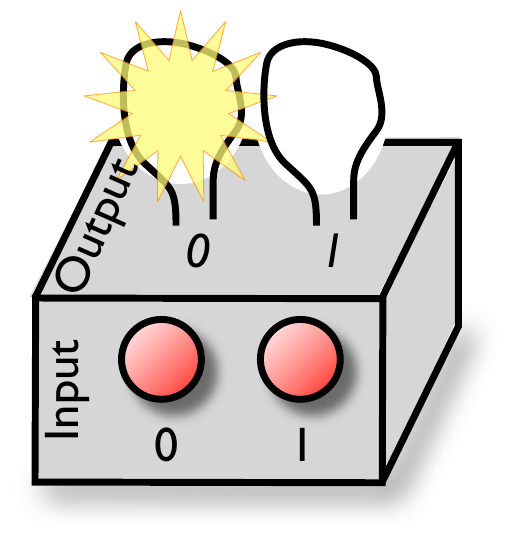}}
\caption{\small 
{\bf Classical interaction with a quantum system.} A~general system can be abstracted as a black box, with two buttons for accepting binary input and two light bulbs for output.  Using this interface, we wish to control fully the system's quantum dynamics.} \label{f:blackbox}
\end{figure}

Although a philosophical question, a positive resolution would have important consequences.  It is particularly relevant in quantum cryptography, where it is natural to model the quantum system as adversarial since the goal is to protect honest users from malicious adversaries.  The \emph{raison d'\^etre} of quantum cryptography is to create a cryptographic system with security premised solely on basic laws of physics, and with quantum key distribution (QKD) and its security proofs~\cite{BennettBrassard84qkd, LoChau98qkdsecurity, ShorPreskill00qkdsecurity} it appeared to have achieved exactly this. However, attackers have repeatedly breached the security of QKD experiments, by exploiting imperfect implementations of the quantum devices~\cite{ZhaoFungQiChenLo07qkdattack, LydersenWiechersWittmannElserSkaarMakarov10qkdattack, GerhardtLiuLamasLinaresSkaarKurtsieferMakarov10qkdbroken}.  Rather than relying on ad hoc countermeasures, Mayers and Yao's $1998$ vision of {device-independent} (DI) QKD~\cite{MayersYao98chsh}, hinted at earlier by Ekert~\cite{Ekert91qkd}, relaxes all modeling assumptions on the devices, and even allows for them to have been constructed by an adversary.  It instead imagines giving the devices tests that cannot be passed unless they carry out the QKD protocol securely.  The challenge at the heart of this vision is for a classical experimentalist to force untrusted quantum devices to act according to certain specifications.  DIQKD has not been known to be possible; security proofs to date require the unrealistic assumption that the devices have no memory between trials, or that each party has many, strictly isolated devices~\cite{BarrettHardyKent04diqkd, MasanesRennerChristandlWinterBarrett06DIQKDnosignalingcomposablesecurity, AcinMassarPironio06DIQKDnosignaling, Masanes09diqkdnosignalingcomposablesecurity, HanggiRennerWolf10diqkd, AcinBrunnerGisinMassarPironioScarani07diqkdcollectiveattacks, PironioAcinBrunnerGisinMassarScarani09qkd, McKague09deviceindependent, HanggiRenner10deviceindependent, MasanesPironioAcin10deviceindependent}.  A scheme for characterizing and commanding a black-box quantum device would provide a novel approach to achieving DIQKD.  

Further, as the power of quantum mechanics is harnessed at larger scales, for example with the advent of quantum computers, it will be useful to evaluate whether a quantum device in fact carries out the claimed dynamics~\cite{AharonovBenOrEban08authenticated, BroadbentFitzsimonsKashefi08authenticated}.  Finally, we might wish to test the applicability of quantum mechanics for large systems, a situation in which Nature itself plays the role of the adversary~\cite{AharonovVazirani12quantummechanics}.

The existence of a general scheme for commanding an adversarial quantum device appears singularly implausible.  For example, in an adversarial setting, experiments cannot be repeated exactly to gather statistics, since a system with memory could deliberately deceive the experimentalist.  More fundamentally, as macroscopic, classical entities, our access to a quantum system is extremely limited and indirect, and the measurements we apply collapse the quantum state.  Furthermore, whereas the dimension of the underlying Hilbert space scales exponentially in the number of particles or can be infinite, the information accessible via measurement only grows linearly. Indeed, as formulated it is impossible to command a single black-box system.  Quite simply, one cannot distinguish between a quantum system that evolves as desired and a device that merely simulates the desired evolution using a classical~computer.  

In this paper, we consider a closely related scenario.  Suppose we are instead given two devices, each modeled as a black box as above, and prevented from communicating with each other.  In this setting, with no further assumptions, we show how to classically command the devices.  That is, there is a strategy for pushing the buttons such that the answering light bulb flashes will satisfy a prescribed test only if the two devices started in a particular initial quantum state, to which they applied a desired sequence of quantum gates.  Moreover, though impractical, the scheme is theoretically efficient---in the sense that the total effort, measured by the number of button pushes, scales as a polynomial function of the size of the desired quantum~circuit.  A DIQKD scheme follows, among other consequences.  

\begin{figure}
\centering
\raisebox{-.25cm}{\includegraphics[scale=.3]{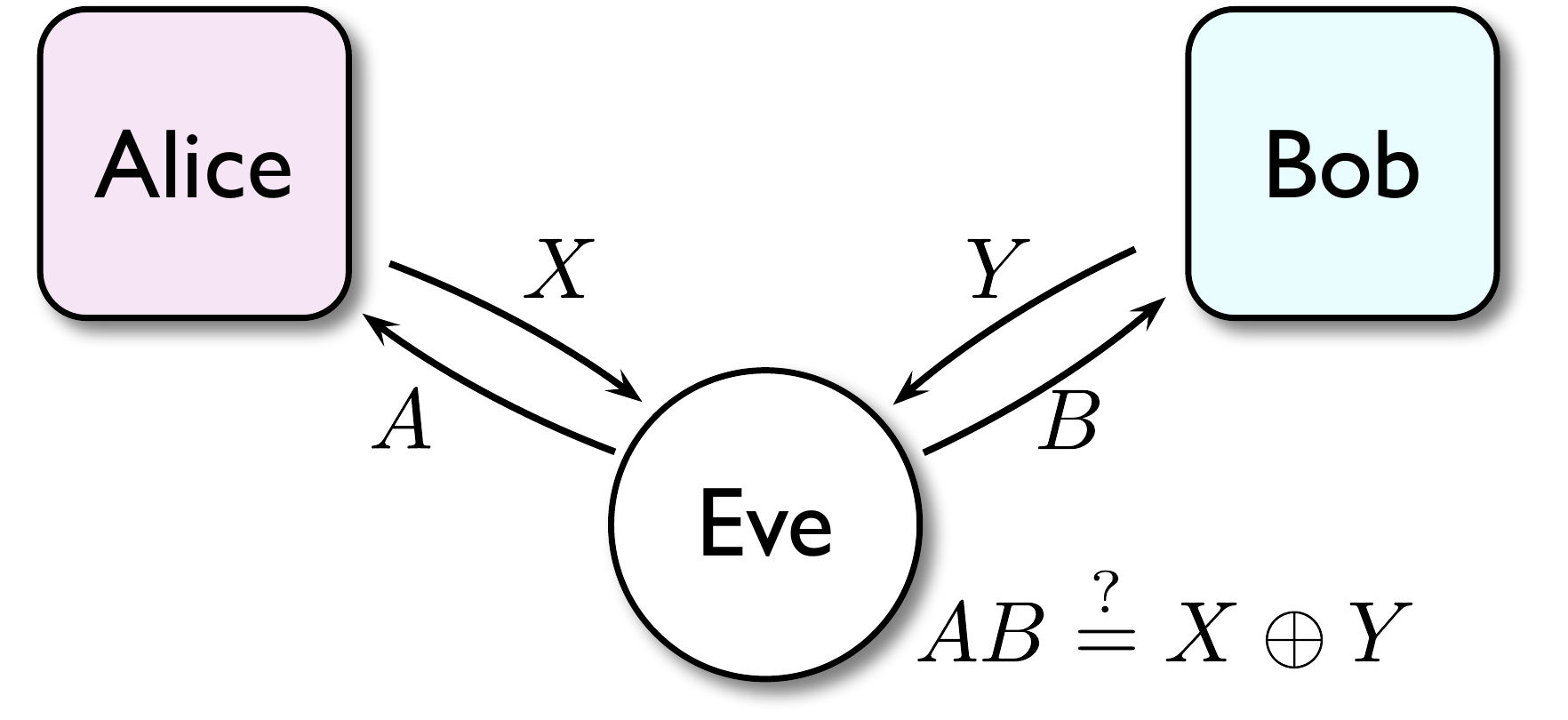}}
\caption{
\small 
{\bf Test for quantumness.}  In a CHSH experiment, or ``game," the experimentalist Eve sends independent, uniformly random bits $A$ and $B$ to the devices Alice and Bob, respectively, who respond with bits $X$ and~$Y$.  The devices ``win" the game if $A B = X \oplus Y$.  By a Bell inequality, classical devices can win with probability at most $3/4$.  Quantum devices can win with probability $\omega^* = \cos^2(\frac\pi8) \approx 85.4\%$, if they follow an ideal CHSH strategy: on a shared Einstein-Podolsky-Rosen (EPR) state $\ket \varphi = \tfrac{1}{\sqrt 2}(\ket{00} + \ket{11})$, Bob measures the Pauli operator $\sigma_z$ if $B = 0$ or $\sigma_x$ if $B = 1$, and Alice measures $\tfrac{1}{\sqrt 2}(\sigma_z + (-1)^A \sigma_x)$.  
Tsirelson showed that this strategy is optimal~\cite{Tsirelson80inequality}.
} \label{f:chsh}
\end{figure}

The starting point for our protocol is the famous Bell experiment~\cite{Bell64epr}, and its subsequent distillation by Clauser, Horne, Shimony and Holt (CHSH)~\cite{ClauserHorneShimonyHolt69chshgame}.  Conceptually modeled as a game (\figref{f:chsh}), it provides a ``test for quantumness," a way for an experimentalist, whom we shall call Eve, to demonstrate the entanglement of two space-like separated devices, Alice and~Bob.  

Consider a protocol in which Eve plays a long sequence of CHSH games with Alice and Bob, and tests that they win close to the optimal fraction $\omega^*$ of the games.  This paper's main technical result establishes that if the devices pass Eve's test with high probability, then at the beginning of a randomly chosen long subsequence of games, Alice and Bob must share many EPR states in tensor product, that they measure one at a time using the single-game ideal CHSH operators of \figref{f:chsh}.  This is a step towards the general vision outlined above because it characterizes the initial state of many qubits, and allows Eve to command the devices to perform certain single-qubit operations.  Of course, we cannot hope to characterize the devices' strategies exactly, but only for a suitable notion of approximation.  

In order to make a more precise statement, first consider a single CHSH game.  We show that if the devices win with probability $\omega^* - \epsilon$, then they must share a state that is $O(\sqrt \epsilon)$-close to an EPR state, possibly in tensor product with an additional state.  Moreover their joint measurement strategy is necessarily $O(\sqrt \epsilon)$-close to the ideal strategy from \figref{f:chsh}.  (That is, applying Alice's actual measurement operator to the shared state gets within distance $O(\sqrt \epsilon)$ of the result from applying her ideal measurement operator to the EPR state tensored with the ancilla; and similarly for Bob.)  Since each device can locate its qubit share of the EPR state arbitrarily within its Hilbert space, these statements hold only up to local isometries.  This may be seen as a robust converse to Tsirelson's inequality, and as a rigidity property of the CHSH game: a nearly maximal Bell inequality violation rigidly locks into place the devices' shared state and measurement directions.  

A converse to Tsirelson's inequality for the CHSH game has been shown previously in the exact case~\cite{BraunsteinRevzen92tsirelsonconverse, PopescuRohrlich92tsirelsonconverse}.  Robustness is important for applications, however, because the success probability of a system can never be known exactly.  Robust, $\epsilon > 0$, converse statements have been shown based on a conjecture~\cite{BardynLiewMassarMcKagueScarani09deviceindependent} or under restrictive symmetry assumptions~\cite{AcinBrunnerGisinMassarPironioScarani07diqkdcollectiveattacks, PironioAcinBrunnerGisinMassarScarani09qkd}.  Recently, robustness has independently been shown for the CHSH game~\cite{McKagueYangScarani12chshrigidity, MillerShi12chshrigidity}.  

Scaling up to a sequence of $n$ CHSH games, suppose Alice and Bob use a strategy such that they win at least $(1 - \epsilon) \omega^* n$ of the games with high probability. By basic statistics, their strategy at the beginning of most games will win that single game with probability at least $(1 - \epsilon^{\Omega(1)}) \omega^*$.  Rigidity for the one-shot game therefore applies.  However, their strategy for playing the $j$th game could depend on the previous games.  The states close to EPR states used in different games could overlap significantly, and their locations could depend on the history.  The multi-game rigidity theorem rules out such wayward behavior.  It says that for most random blocks of $m = n^{\Omega(1)}$ consecutive games, at the start of the block Alice and Bob must share a state that is close to a tensor product of $m$ EPR states, tensored with an additional state, and must play each $j$th game by making measurements that are close to the ideal CHSH strategy on the $j$th EPR state---different games being entirely independent.  

One way to view this theorem is that it scales up the CHSH test for quantumness and allows for identifying many qubits' worth of entanglement.  Much more than that, however, the multi-game rigidity theorem gives strong control over the devices' measurement operators for different games.  As described below, combining the CHSH game protocol with protocols for state and process tomography, and for computation by teleportation, thereby gives a method for realizing arbitrary dynamics in quantum systems without making assumptions about the internal structure or operations.  The dynamics are realized as the joint evolution of two isolated quantum systems, Alice and Bob, mediated by a classical experimentalist, Eve.  

The problem of controlling computationally powerful but untrusted resources lies at the foundation of computer science.  In the complexity class~$\NP$, for example, a polynomial-time routine---the ``verifier"---is allowed one round of interaction with an arbitrarily powerful, but malicious, ``prover."  We show that the same verifier can exploit the power of quantum-mechanical provers.  In particular: 

1. A classical verifier can efficiently simulate a quantum computer by interacting with two untrusted, polynomial-time quantum provers.  This delegated computation scheme is also \emph{blind}, meaning that each prover learns no more than the length of the computation.  

2. The verifier in any quantum multi-prover interactive proof system can be assumed to be classical.  Formally, the complexity classes $\QMIP$ and $\MIP^*$ are equal.  

\noindent
Previous work has considered a verifier who can store and control a constant number of qubits while interacting with a single prover~\cite{AharonovBenOrEban08authenticated, BroadbentFitzsimonsKashefi08authenticated, FitzsimonsKashefi12blind, BarzKashefiBroadbentFitzsimonsZeilingerWalther12blindexperiment}.  Our work is also inspired by a proposal~\cite{BroadbentFitzsimonsKashefi10qmip} that $\QMIP$ should equal $\MIP^*$.  Our protocol has a very different form, based on the multi-game rigidity theorem.  

%\vspace{-.08in}
\subsubsection*{Tensor-product structure for repeated %CHSH 
games} \label{s:sequentialstructuredCHSHgameshavetensorproductstructuresketch}
%\vspace{-.06in}

A strategy $\S$ for playing $n$ sequential CHSH games specifies Alice and Bob's initial joint state as well as their measurement operators for every possible situation.  That is, for $X \in \{A, B\}$ and each $j = 1, \ldots, n$, $\S$ specifies the measurement operators used by device~$X$ in game $(j, \hX{j-1})$, where $\hX{j-1}$ is any transcript of the device's input and output bits for the first $j-1$ games.  For two strategies to be ``close" means that the distributions of game transcripts they induce should be close in total variation distance; and that for most transcripts (drawn from either distribution), the resulting quantum states should be close in a suitable norm.  We combine these conditions into one by defining for any strategy a block-diagonal density matrix that stores both the classical transcript and the resulting quantum~state: 
\begin{equation} \label{e:examplecombinedtranscriptstatedensitymatrix}
\rhoj{j} = \bigoplus_{h_{j-1}} \Pr[h_{j-1}] \, \rhoj{j}(\h{j-1})
 \enspace .
\end{equation}
Here $\h{j-1} = (\hA{j-1}, \hB{j-1})$ is the full transcript for the first $j-1$ games and $\rhoj{j}(\h{j-1})$ is the state at the beginning of game~$j$ conditioned on $\h{j-1}$.  Two strategies~$\S$ and~$\tilde \S$ are close if the~associated $\rhoj{j}$ and $\rhodecj{\tilde}{j}$ are close in trace distance, for every~$j$.  

Assume that for every $j$ and most $\h{j-1}$, the devices' conditional joint strategy at the beginning of game~$j$ is ``$\epsilon$-structured," meaning that it wins with probability at least $\omega^* - \epsilon$.  Our key theorem establishes that up to local basis changes, the devices' initial state must be close to $n$ EPR states, possibly in tensor product with an irrelevant extra state, and that their total strategy $\S$ must be close to an ideal strategy $\hat{\S}$ that plays game~$j$ using the $j$th EPR state.  Since the structure assumption can be established by martingale arguments on $\poly(n)$ sequential CHSH games, this implies the multi-game rigidity theorem.  

The main challenge is to ``locate" the ideal strategy $\hat \S$ within Alice and Bob's Hilbert space, i.e., to find an isometry on each of their spaces under which their states and measurement operators are close to ideal.  However, a priori, we do not know whether $\S$ calls for the devices to measure actual qubits in each step, or even if so whether the qubits form EPR states, qubits for different games overlap each other, or the locations of the qubits depend on the outcomes of previous games.  

The construction to locate the qubits in the single-game rigidity theorem is a good~place to start.  Consider an $\epsilon$-structured strategy, consisting of some shared mixed state in $\H_A \otimes \H_B$, and two-outcome projective measurements for each of Eve's possible questions.  Truncate the devices' Hilbert spaces to finitely many dimensions, then decompose each space by Jordan's lemma~\cite{Jordan75projections} into the direct sum of two-dimensional spaces invariant under the projections.  Within each such two-dimensional subspace adjust the angles between the projections to match the ideal strategy.  The resulting operators define underlying qubits.\footnote{A %longer 
technical account with full proofs will appear elsewhere~{\cite{ReichardtUngerVazirani12qmip}.}}  

For multiple CHSH games, the given strategy~$\S$ can be transformed into a nearby ideal strategy $\hat \S$ by a three-step sequence:

1. First, replace each device's measurement operators by the ideal operators promised by the single-game rigidity theorem.  In the resulting strategy $\tilde \S$, each device $X$ plays every game~$(j, \hX{j-1})$ using the ideal CHSH game operators on some qubit, up to a local change in basis.  However, the basis change can depend arbitrarily on $\hX{j-1}$, and the qubits for different~$j$ need not be in tensor product.  

2. In a ``multi-qubit ideal strategy" $\bar \S$, the qubits used in each game can still depend on the local transcripts but must at least lie in tensor product with the qubits from previous games. This imposes a tensor-product subsystem structure that previous DIQKD proofs have assumed.  The tensor-product structure is constructed beginning with a trivial transformation on~$\tilde \S$: to each device, add~$n$ ancilla qubits each in state~$\ket 0$.  Next, after a qubit has been measured, say as~$\ket{\alpha_j}$ in game~$j$, swap it with the $j$th ancilla qubit, then rotate this fresh qubit from $\ket 0$ to $\ket{\alpha_j}$ and continue playing games $j+1, \ldots, n$.  This defines a unitary change of basis that places the outcomes for games $1$ to~$j$ in the first $j$ ancilla qubits, and leaves the state in the original Hilbert space unchanged.  Since qubits are set aside after being measured, the qubits for later games are automatically in tensor product with those for earlier games; the resulting strategy $\bar \S$ is multi-qubit ideal.  At the end of the $n$ games, swap back the ancilla qubits and undo their rotations, using the transcript.  

3. In the last step, we replace $\bar \S$ with an ideal strategy $\hat{\S}$, in which Alice and~Bob each play using a fixed set of $n$ qubits.  Fix a transcript $\hdec{\hat}{n}$, chosen at random.  For the first time, change the devices' initial state: replace $\rhoone$ with $\rhodecone{\hat}$, a state having $n$ EPR states in the locations determined by $\hdec{\hat}{n}$ in $\bar \S$.  In~$\hat \S$, the devices play using these EPR states, regardless of the actual transcript.  This $\hat \S$ is the desired ideal strategy.

%\vspace{-.08in}
\subsubsection*{Ideal strategy $\hat S$ is close to $\S$} \label{s:close}
%\vspace{-.06in}

It remains to show that the transformation's three steps incur a small error: $\hat \S$ is close to $\S$.  A major theme in the analysis is to leverage the known tensor-product structure between $\H_A$ and $\H_B$ to extract a tensor-product structure within $\H_A$ and~$\H_B$.

1. $\S \approx \tilde \S$: Although elementary, explaining this step is useful for establishing some notation.  Let $\rhoone$ be the devices' initial shared state, possibly entangled with the environment.  Let~$\EAj{j}$ and $\EBj{j}$ be the super-operators that implement Alice and Bob's respective strategies for game~$j$, $\EABj{j} = \EAj{j} \otimes \EBj{j}$ and $\EABj{j,k} = \EABj{k} \cdots \EABj{j}$ for $j \leq k$; thus the state $\rhoj{j}$ of Eq.~\eqnref{e:examplecombinedtranscriptstatedensitymatrix} equals $\EABj{1,j-1}(\rhoone)$.  For $D \in \{A, B\}$, let $\EXdecj{\tilde}{j}$ be the super-operator that replaces the actual measurement operators with the ideal operators promised by the CHSH rigidity theorem.  $\tilde \S$ is given by $\rhoone$, $\{ \EAdecj{\tilde}{j} \}$ and $\{ \EBdecj{\tilde}{j} \}$.  If $\Pr[\text{game~$j$ is $\epsilon$-structured}] \geq 1-\delta$, then $\trnorm{\EABj{j}(\rhoj{j}) - \EABdecj{\tilde}{j}(\rhoj{j})} \leq 2 \delta + O(\sqrt \epsilon)$.  (This expression uses Eq.~\eqnref{e:examplecombinedtranscriptstatedensitymatrix} to combine bounds on the probability of the bad event and the $O(\sqrt \epsilon)$ error from the good event.)  To show our goal, that $\EABj{1,n}(\rhoone) \approx \EABdecj{\tilde}{1,n}(\rhoone)$ in trace distance, use a hybrid argument that works backwards from game~$n$ to game~$1$ fixing each game's measurement operators one at a time.  

2. $\tilde \S \approx \bar \S$: The key to showing that $\bar \S$ is close to $\tilde \S$ is the fact that operations on one half of an EPR state can equivalently be performed on the other half, since for any $2 \times 2$ matrix~$M$, $(M \otimes I)(\ket{00} + \ket{11}) = (I \otimes M^T)(\ket{00} + \ket{11})$.  This means that the outcome of an $\epsilon$-structured CHSH game would be nearly unchanged if Bob were hypothetically to perform Alice's measurement before his own.  By moving Alice's measurement operators for games $j+1$ to $n$ over to Bob's side, we see that they cannot significantly affect the qubit $\ket{\alpha_j}$ from game~$j$ on her side.  Therefore, undoing the original change of basis restores the ancilla qubits nearly to their initial state $\ket{0^n}$, and $\tilde \S \approx \bar \S$.  

\def\AmeasBBdecj #1#2{#1{\F}^{AB}_{#2}}

Formally, define a unitary super-operator ${\cal V}_j$ that rotates the $j$th ancilla qubit to $\ket{\alpha_j}$, depending on Alice's transcript~$\hA{j}$.  Define a unitary super-operator ${\cal T}_j$ to apply ${\cal V}_j$ and swap the $j$th ancilla qubit with the qubit Alice uses in game~$j$ (depending on~$\hA{j-1}$).  Alice's multi-qubit ideal strategy is given~by 
\begin{equation}
\EAdecj{\bar}{j} = {\cal T}_{1,j-1}^{-1} (\identity_{\C^{2^n}} \otimes \EAdecj{\tilde}{j}) {\cal T}_{1,j-1}
 \enspace .
\end{equation}
We aim to show that the strategy given by $\rhoone$, $\{ \EAdecj{\bar}{j} \}$ and $\{ \EBdecj{\tilde}{j} \}$ is close to~$\tilde \S$ up to the fixed isometry that prepends $\ketbra{0^n}{0^n}$ to the state, i.e., that $\ketbra{0^n}{0^n} \otimes \EABdecj{\tilde}{1,n}(\rhoone) \approx \EAdecj{\bar}{1,n}\big( \ketbra{0^n}{0^n} \otimes \EBdecj{\tilde}{1,n} (\rhoone) \big)$.  Define a super-operator $\AmeasBBdecj{\tilde}{j}$, in which Alice's measurements are made on \emph{Bob's} Hilbert space~$\H_B$, on the qubit determined by Bob's transcript $\hB{j-1}$.  Since most games are $\epsilon$-structured, by the CHSH rigidity theorem, $\AmeasBBdecj{\tilde}{j+1,k}(\rhodecj{\tilde}{j+1}) \approx \EABdecj{\tilde}{j+1,k}(\rhodecj{\tilde}{j+1}) = \rhodecj{\tilde}{k+1}$ for $j \leq k$.  Since $\AmeasBBdecj{\tilde}{j+1,k}$ acts on $\H_B$, it does not affect Alice's qubit $\ket{\alpha_j}$ from game~$j$ at all, and so this qubit must stay~near $\ket{\alpha_j}$ in $\rhodecj{\tilde}{k+1}$ as well, i.e., the trace of the reduced density matrix against the projection $\ketbra{\alpha_j}{\alpha_j}$ stays close to one.  As this holds for every~$j$, ${\cal T}_{1,n}^{-1}$ indeed returns the ancillas almost to their~initial state~$\ket{0^n}$.  The $\{ \EBdecj{\tilde}{j} \}$ are symmetrically adjusted to $\{ \EBdecj{\bar}{j} \}$.  

3. $\bar \S \approx \hat \S$: In $\bar \S$, Alice and Bob play according to a strategy in which every game uses a qubit in tensor product with the previous games' qubits.  However, the qubit's location can depend on previous games' outcomes.  We wish to argue that Alice and Bob must play using a single set of $n$ qubits, fixed in advance independent of the transcript.  

Intuitively, if the location of Alice's $j$th qubit depended on $\hA{j-1}$, then since the devices cannot communicate with each other, Bob could not know which of his qubits to measure.  However, Alice and Bob's transcripts are significantly correlated, and we must show that they cannot use these correlations to coordinate dynamically the locations of their qubits. 

For a toy example that illustrates the issue, consider two devices who play the first $n-1$ games honestly and which at the beginning of the last game share two EPR states, $\ket{\varphi}^{\otimes 2}$.  Say that for certain functions~$f$ and~$g$, Alice uses EPR state $f(\hA{n-1}) \in \{1,2\}$ in game~$n$, and Bob uses EPR state $g(\hB{n-1}) \in \{1,2\}$.  For game~$n$ to be structured, they need $f(\hA{n-1}) = g(\hB{n-1})$ so that they measure the same EPR state.  Now Alice and Bob's local transcripts are each uniformly random, separately, but they have a constant correlation in every game coordinate.  It is straightforward to argue based on coordinate influence that if $\Pr[f(H^A_{n-1}) \neq g(H^B_{n-1})]$ is small, then~$f$ and~$g$ must both be nearly constant.  Thus one of the two EPR states is used almost always.  

This example gives an essentially classical cheating strategy.  The actual devices may be significantly more sophisticated.  In particular, small amounts of cheating in earlier games might enable an avalanche of more and more blatant cheating in later games, drastically changing the underlying quantum state.  If, for example, Alice knowingly manages to swap her halves of the two last EPR states along some transcripts~$\hA{n-1}$, then she can use completely different strategies for the last game without having to coordinate with Bob.  We control such errors, like in the arguments sketched above, by approximating Alice's super-operator as acting on Bob's Hilbert space.   More formal arguments are deferred to the supplementary material.  

\smallskip

The conclusion that the devices' joint strategy is close to ideal, $\EABj{1,n}(\rhoone) \approx \EABdecj{\hat}{1,n}(\rhodecone{\hat})$, is not strong enough for our applications, in which sometimes Eve plays CHSH games with only one of the two devices.  We need to show that the devices' strategies are \emph{separately} close to ideal, i.e., $\EAj{1,n}(\rhoone) \approx \EAdecj{\hat}{1,n}(\rhodecone{\hat})$ and $\EBj{1,n}(\rhoone) \approx \EBdecj{\hat}{1,n}(\rhodecone{\hat})$.  These estimates cannot be obtained directly because our main assumption, that every game~$j$ is usually $\epsilon$-structured, is only of use if both devices have played games~$1$ through~$j-1$---it gives information about $\EXj{j}$ applied to $\EABj{1,j-1}(\rhoone)$, not about~$\EXj{j}$ applied to $\EXj{1,j-1}(\rhoone)$.  The key idea to obtain separate estimates is that applying both devices' super-operators is almost equivalent to applying Alice's super-operator, \emph{guessing} Bob's measurement outcome from the ideal conditional distribution, and based on the guess applying a controlled unitary correction to his qubit.  Since Alice's super-operator collapses both qubits of the EPR state, it is not actually necessary to measure Bob's qubit.  Defining $\Bguessj{j}$ to be this guess-and-correct super-operator, two hybrid arguments give $\EABj{1,n}(\rhoone) \approx \Bguessj{1,n} \EAj{1,n}(\rhoone)$ and $\EAdecj{\tilde}{1,n} \EBj{1,n}(\rhoone) \approx \Bguessj{1,n} \EAdecj{\tilde}{1,n}(\rhoone)$.  Thus, 
\begin{equation*}
\Bguessj{1,n} \EAj{1,n}(\rhoone) \approx \Bguessj{1,n} \EAdecj{\tilde}{1,n}(\rhoone)
 \enspace .
\end{equation*}
The same super-operator $\Bguessj{1,n}$ appears on both the left- and right-hand sides above.  In general, applying a super-operator can reduce the trace distance.  In this case, however, it does not; the correction part of $\Bguessj{1,n}$ is unitary, and the guessing part is a stochastic map acting on a {copy} of Alice's classical transcript register.  Therefore, indeed $\EAj{1,n}(\rhoone) \approx \EAdecj{\tilde}{1,n}(\rhoone)$.  The third step of the proof uses a similar, but more involved, argument.

%\vspace{-.08in}
\subsubsection*{Verified quantum dynamics} \label{s:tomographysketch} 
%\vspace{-.06in}

\begin{figure*}[t]
\centering
\subfigure[\label{f:chshgames} CHSH games]{$\!\!\!\!\!\!$\includegraphics[scale=.5]{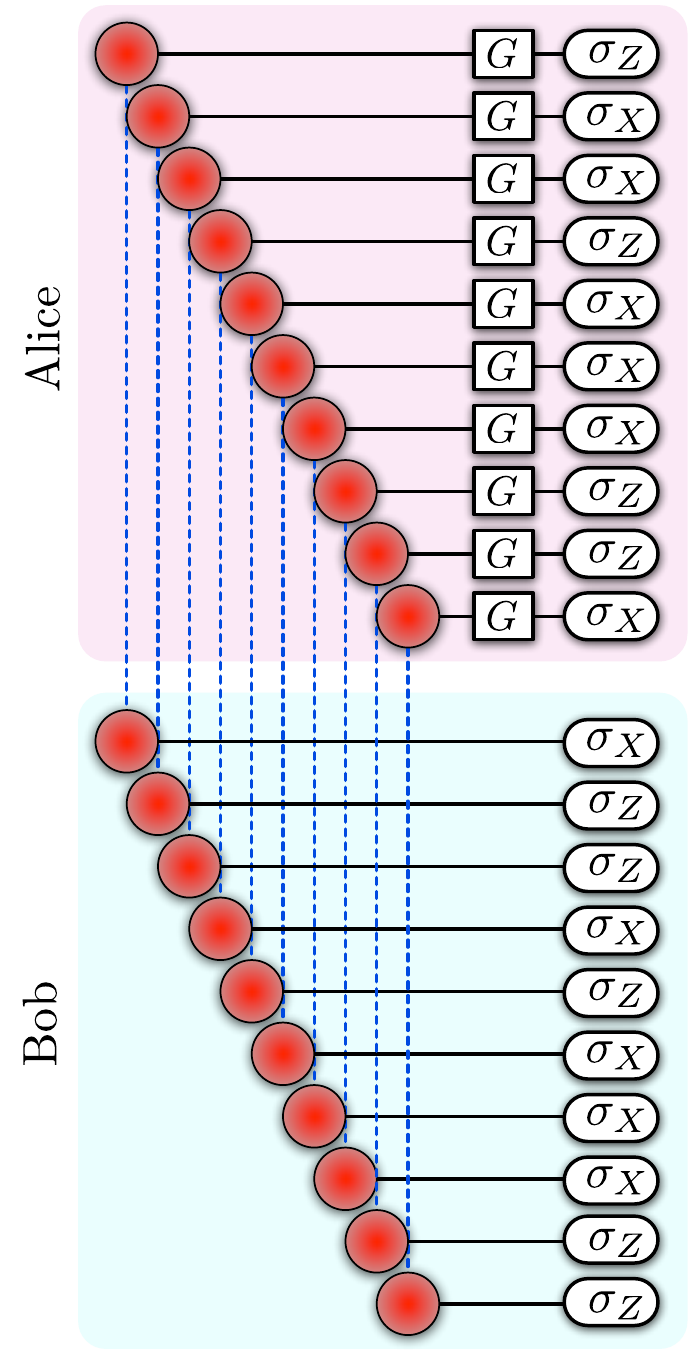}} 
$\,$
\subfigure[\label{f:statetomography} State tomography]{\includegraphics[scale=.5]{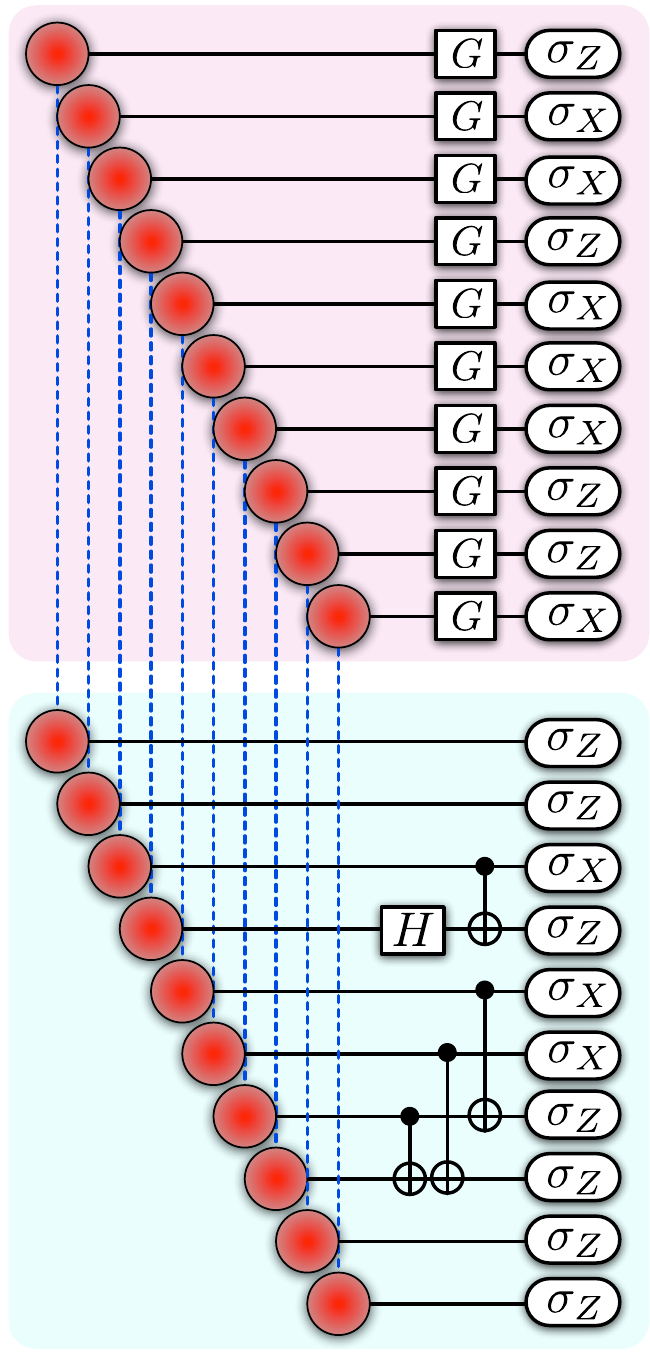}}
$\,$
\subfigure[\label{f:processtomography} Process tomography]{\includegraphics[scale=.5]{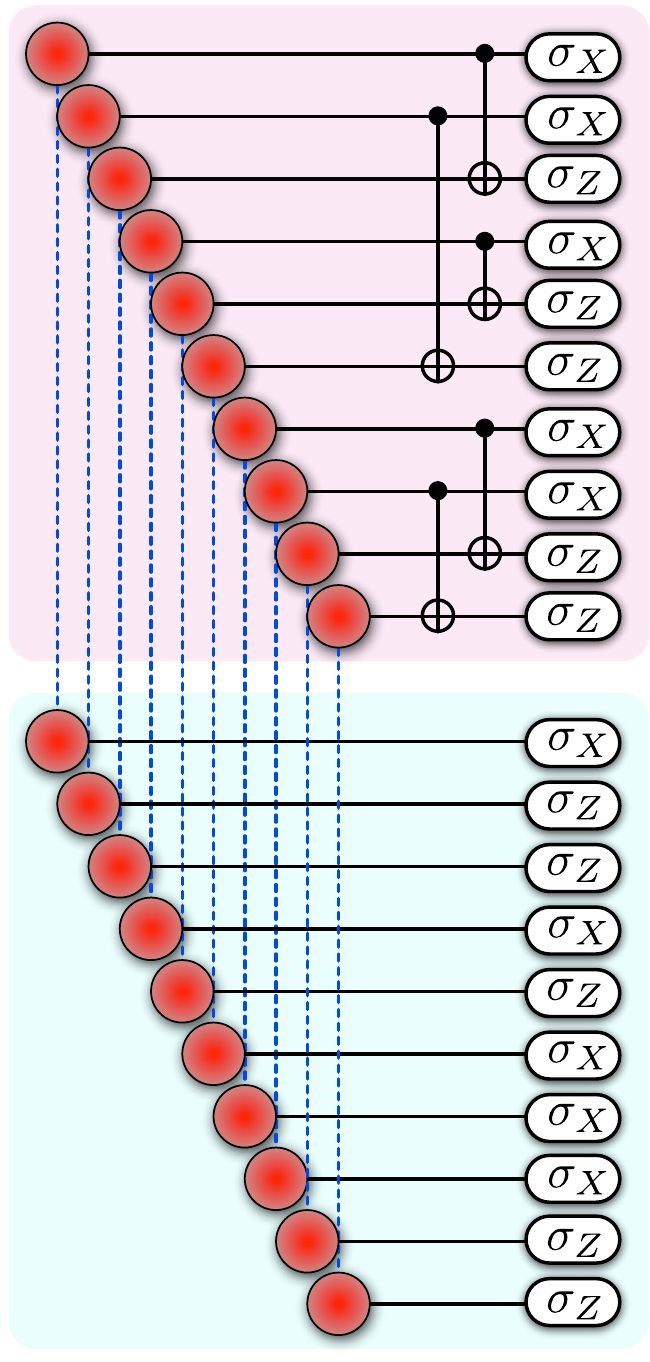}}
$\,$
\subfigure[\label{f:computation} Computation]{\includegraphics[scale=.5]{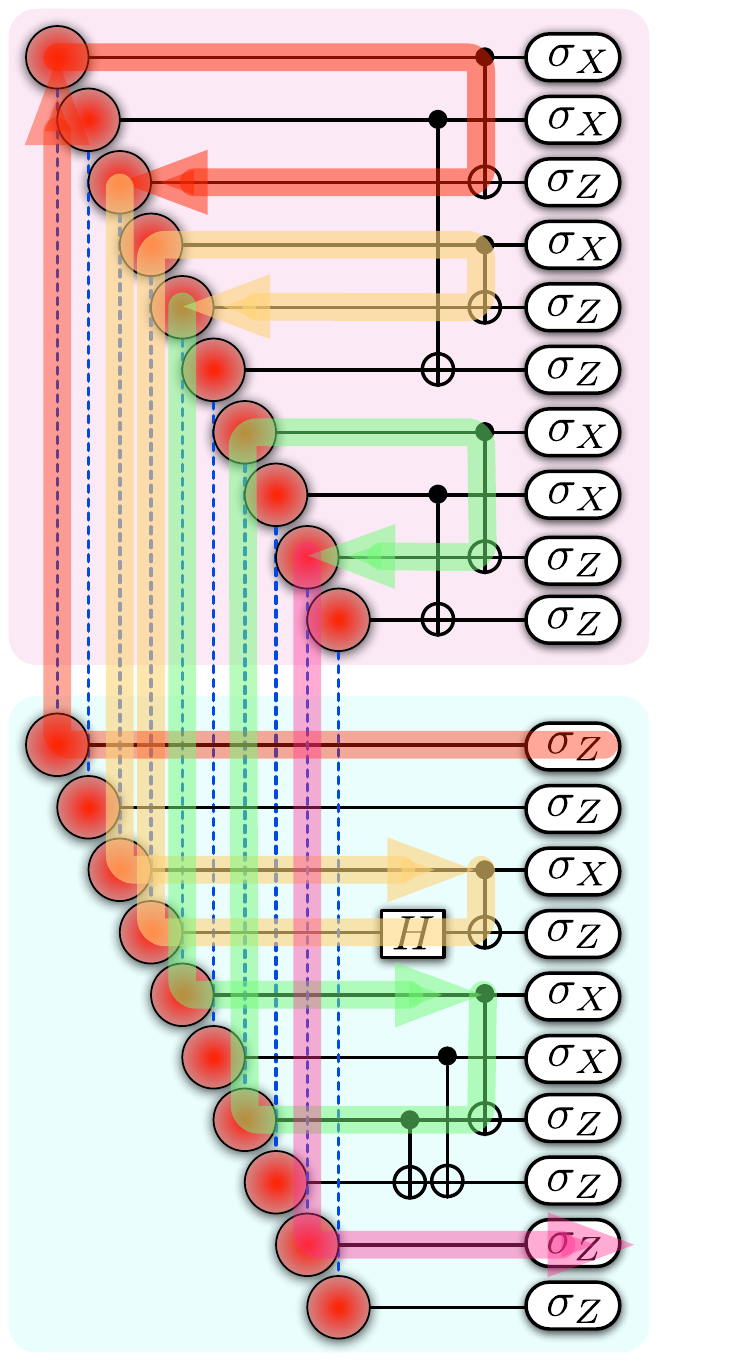}$\!\!\!\!\!\!\!\!$}
\raisebox{2.53cm}{\begin{tabular}{c}
\subfigure[\label{f:alicebob} Circuit $\mathcal C$]{\raisebox{.5cm}{\includegraphics[scale=.7]{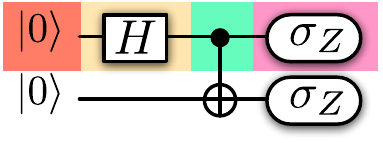}}} \\[.6cm]
\subfigure[\label{f:teleportation} Teleporting into $H$]{\raisebox{.7cm}{$\!\!\!\!\!\!\!$\includegraphics[scale=.7]{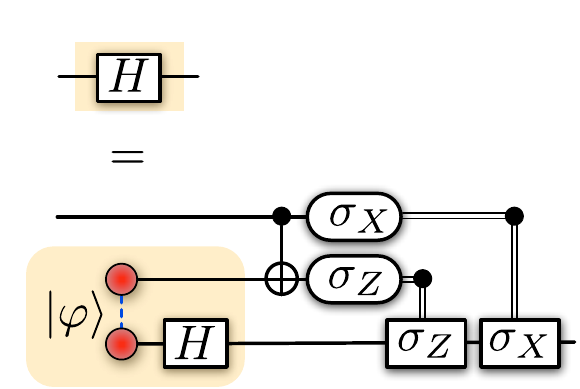}$\!\!\!\!\!\!\!\!\!\!\!\!$}}
\end{tabular}}
\caption{\small
{\bf Sub-protocols for verified quantum dynamics.}  
To delegate a quantum computation, Eve runs a random one of four sub-protocols with Alice and Bob.  {\bf a}, Playing many CHSH games ensures that the devices play honestly using shared EPR states.  {\bf b-c}, This lets Eve apply state or process tomography to characterize more complicated multi-qubit operations.  {\bf d-e}, By adaptively combining these operations, Eve directs a quantum circuit~$\mathcal C$.  The zig-zagging logical path of the first qubit of~$\mathcal C$ is highlighted.  {\bf f}, Each gate of $\mathcal C$ is implemented through teleportation; in this simpler example, $H$ is applied by a Bell measurement on half of the resource state $(I \otimes H) \ket \varphi$.    
} \label{f:computationprotocols}
\end{figure*}

Our scheme for verified quantum dynamics is based on the~idea of computation by teleportation, which reduces computation to preparing certain resource states and applying Bell measurements (\figref{f:teleportation})~\cite{GottesmanChuang99teleportation}.  Say that Eve wants to~simulate a quantum circuit~$\mathcal C$, over the gate set $\{ H, G, \mathrm{CNOT} \}$,~where $H$ is the Hadamard gate and $G = \exp(-i \frac{\pi}{8} \sigma_y)$.  Eve asks Bob to prepare many copies of the state $\ket 0 \otimes (I \otimes H) \ket{\varphi} \otimes (I \otimes \phasegate) \ket{\varphi} \otimes \mathrm{CNOT}_{2,4} (\ket{\varphi} \otimes \ket{\varphi})$.  He can do so by applying one-, two- and four-qubit measurements to his halves of the~shared EPR states and reporting the results to Eve.  If he plays honestly, Alice's shares of the EPR states collapse into the desired resource states, up to simple corrections.  Each resource state corresponds to a basic operation in~$\mathcal C$.  Eve wires these~up by repeatedly directing Alice to make a Bell measurement connecting the output of one operation to the input of the next~operation in~$\mathcal C$.  After each $G$ gate, an $H$ correction might be required.  

Of course, Alice and Bob might not follow directions.  To enforce honest play, Eve runs this protocol only a small fraction of the time, and otherwise chooses uniformly between three alternative protocols sketched in \figref{f:computationprotocols}.  Let $m = {\abs C}^{O(1)}$ and $n = m^{O(1)}$.  

1. In the ``state tomography" protocol, Eve chooses $K$ uniformly from $\{1, \ldots, n/m\}$.  She referees $(K-1) m$ CHSH games.  Then in the $K$th block of~$m$, Eve asks Bob to prepare the resource states, in a random order, while continuing to play CHSH games with Alice.  Eve rejects if the tomography statistics are inconsistent.  We prove that if Alice plays honestly and Eve accepts with high probability, then on most randomly chosen small subsets of the resource state positions, Alice's reduced state is close to the correct tensor product of resource states.  

2. In the ``process tomography" protocol, Eve again chooses $K$ uniformly from $\{1, \ldots, n/m\}$ and referees $(K-1) m$ CHSH games.  In the $K$th block of~$m$, Eve asks Alice to make Bell measurements on random pairs of qubits, while continuing to play CHSH games with Bob.  If Alice's reported result for \emph{any} pair of qubits is inconsistent with Bob's outcomes, Eve rejects.  Then if Bob plays honestly and Eve accepts with high probability, Alice must also have applied the Bell measurements~honestly.  

3. In the third protocol, Eve simply referees $n$ sequential CHSH games with both devices and rejects if they do not win at least $(1 - \epsilon) \omega^* n$ games.   

From Bob's perspective the process tomography and computation protocols are indistinguishable, as are the state tomography and CHSH game protocols.  From Alice's perspective, the state tomography and computation protocols are indistinguishable, as are the process tomography and CHSH game protocols.  The devices must behave identically in indistinguishable protocols.  The multi-game rigidity theorem therefore provides the base for a chain of implications that implies that if Eve accepts with high probability, then the devices must implement $\mathcal C$ honestly.  

Four main technical problems obstruct these claims.  

First, in the state tomography protocol, if Bob is dishonest, then Alice gets an arbitrary $m$-qubit state, and there is no reason why it should split into a tensor product of repeated, constant-qubit states.  Nonetheless, we argue using martingales that if the counts of Alice's different measurement outcomes roughly match  their expectations with high probability, then for most reported measurement outcomes from Bob and for most subsystems~$j$, Alice's conditional state reduced to her $j$th subsystem is close to what it should be.  

Furthermore, saturating Tsirelson's inequality for the CHSH game only implies that Alice is honestly making Pauli $\sigma_x$ and~$\sigma_z$ measurements on her half of an EPR state.  Tomography also requires $\sigma_y$ measurements.  To sidestep this issue, we generalize a theory introduced by McKague~\cite{McKague10thesis} and prove that there is a large class of states, including the necessary resource states, that are all robustly determined by only $\sigma_x$ and~$\sigma_z$ measurements.  

A bigger problem, though, is that we want to characterize the operations that the devices apply to their shared EPR states, and not just the states that these operations create on the other side.  The distinction is the same as that between process and state tomography.  Essentially, the problem is that the correct states could be generated by incorrect processes.  Moreover, as for sequential CHSH games, Bob's strategy in early tomography rounds might be sufficiently dishonest as to allow him in later rounds to apply completely dishonest operators.  A key observation to avoid this problem is that it is enough to certify the states prepared by one device and the processes applied by the other.  Then since a broad class of states can be certified, for applications it suffices to certify a much smaller set of operations.  We restrict consideration to Pauli stabilizer measurements~\cite{Gottesman97thesis}.  For Pauli operators in the stabilizer of a state, the measurement outcome is deterministic.  Therefore if Alice reports the wrong stabilizer syndrome in even a single round, Eve can reject.  Our process certification analysis is similar to some of the arguments used above.  We argue that Alice's earlier measurements cannot usually overly disturb the qubits intended for use in later measurements, by pulling Alice's measurement super-operators over onto Bob's halves of the EPR states.  

Finally, the verifier's questions in the state and process tomography protocols are non-adaptive, whereas in computation by teleportation the questions must be chosen adaptively based on previous responses.  This is an attack vector in some related protocols.  However, we argue that the devices can learn nothing from the adaptive questions.  This follows because computation by teleportation can be implemented exactly equivalently either by choosing Bob's state preparation questions non-adaptively and Alice's process questions adaptively, or vice versa.  

The proof that $\QMIP = \MIP^*$ follows along similar lines.  Begin with a $k$-prover protocol.  We may assume that it has two rounds of quantum messages from the provers, before and after the verifier broadcasts a random bit~\cite{KempeKobayashiMatsumotoVidick07qmip}.  To convert to an MIP$^*$ protocol, with classical messages, add two additional provers, Alice and Bob.  Eve teleports the original $k$ provers' messages to Alice, and directs Alice and Bob together to apply the quantum verifier's acceptance predicate.

%\vspace{-.08in}
\subsubsection*{Discussion}
%\vspace{-.06in}

By characterizing the device strategies that can win many successive CHSH games, we have shown how a fully classical party can direct the actions of two untrusted quantum devices.  The simplest case is device-independent quantum key distribution, free of the independence assumptions needed in previous analyses.  A major open problem is to improve the efficiency of our schemes, so as to get a key rate that is a positive constant instead of inverse polynomially small.  More generally, the CHSH game rigidity theorems have no classical analog; and further implications for cryptographic primitives and protocols with drastically reduced security assumptions remain to be explored.

%auto-ignore
\ifx\compilefullpaper\undefined  
\documentclass[twocolumn, 9pt]{extarticle}
\usepackage[superscript]{cite}
\usepackage{enumitem}

% This is code to compress the bibliography slightly.  
\let\oldthebibliography=\thebibliography
\let\endoldthebibliography=\endthebibliography
\renewenvironment{thebibliography}[1]{%
  \begin{oldthebibliography}{#1}%
    \setlength{\parskip}{0ex}%
    \setlength{\itemsep}{0ex}%
}%
{%
  \end{oldthebibliography}%
}

\usepackage[letterpaper]{geometry}		\geometry{includefoot,verbose,nohead,tmargin=.7in,bmargin=.7in,lmargin=.7in,rmargin=.7in}

\begin{document}

\title{\vspace{-.45in} \Huge Classical command of quantum systems via rigidity of CHSH games: \\ Supplementary Information}
\author{Ben W.~Reichardt \and Falk Unger \and Umesh Vazirani}
\date{}

\maketitle
\else 
\appendix
\fi

\section{Rigidity for the CHSH game}

%Let us
In this appendix, we sketch the proof of the single-game CHSH rigidity theorem for the case that the devices' Hilbert spaces $\H_A$ and $\H_B$ are finite dimensional, and $\epsilon = 0$, i.e., the devices achieve the maximum correlation allowed by Tsirelson's inequality.  Hilbert space completeness allows for truncating an infinite-dimensional space to finitely many dimensions, at an arbitrarily small~cost.  

A general strategy for the devices consists of measuring some shared state in $\H_A \otimes \H_B$.  Since the success probability is extremal, i.e., $\epsilon = 0$, we may assume that the state is extremal, i.e., is pure.  Each device measures its system using a reflection that depends on Eve's question, and returns the sign of the observed eigenvalue $\pm 1$.  The shared state $\ket \psi$ and Alice and Bob's four reflections $\RAa{0}$, $\RAa{1}$, $\RBa{0}$ and $\RBa{1}$ determine the strategy.  

Jordan's Lemma~\cite{Jordan75projections} states that any two reflections acting on a finite-dimensional space can be simultaneously block-diagonalized into $1 \times 1$ and $2 \times 2$ blocks.  (The same statement is false for infinite-dimensional spaces.)  Apply the lemma to $\RAa{0}$ and $\RAa{1}$, to obtain 
\begin{equation} \label{e:RAablockdiagonalization}
\RAa{a} 
= \bigoplus_i \RAa{a}(i)
 \enspace ,
\end{equation}
where $i$ labels the block index, and each $\RAa{a}(i)$ is a $1 \times 1$ or $2 \times 2$ reflection.  By adding placeholder dimensions, we may assume without loss of generality that each block is $2 \times 2$.  Eq.~\eqnref{e:RAablockdiagonalization}, which can equivalently be rewritten $\RAa{a} = \sum_i \RAa{a}(i) \otimes \ketbra i i$, gives a basis in which $\H_A = \C^2 \otimes \H_A'$, where $\H_A'$ is the Hilbert space with orthonormal basis $\{ \ket i \}$.  Thus Jordan's Lemma gives (an extension of) the a~priori formless space $\H_A$ a tensor-product structure.  It locates within $\H_A$ a qubit $\C^2$.  However, Alice's operators need not act locally on this qubit; a controlled rotation is still needed to align her operators.  After this rotation, we will show that Alice's strategy is close to the ideal CHSH game strategy that measures this qubit using the operators given in \figref{f:chsh}.  

Since the measurement of the block index commutes with both $\RAa{a}$, we may assume that Alice measures the block index first.  Thus we reduce to the case that $\H_A = \C^2$, and, by a symmetrical argument, that $\H_B = \C^2$, still with $\epsilon = 0$.  

If the $\RXa{\alpha}$ reflections act on $\C^2$ and are not equal to $\pm I$, then we can choose a basis such that $\RXa{0} = \sigma_z = \smatrx{1&0\\0&-1}$, $\RAa{1} = \smatrx{\cos 2 \theta & \sin 2 \theta \\ \sin 2 \theta & -\cos 2 \theta}$ and $\RBa{1} = \smatrx{\cos 2 \theta^{\smash{\prime}} & \sin 2 \theta^{\smash{\prime}} \\ \sin 2 \theta^{\smash{\prime}} & -\cos 2 \theta^{\smash{\prime}}}$ for certain angles $\theta, \theta' \in [0, \frac\pi2]$.  (As this basis depends on the block index, changing basis amounts to a controlled qubit rotation.)  Letting $M_a = \frac12 (\RAa{0} + (-1)^a \RAa{1}) \otimes I - \frac{1}{\sqrt 2} I \otimes \RBa{a}$ for $a \in \{0,1\}$, the success probability satisfies 
\begin{equation*}\begin{split}
2 \sqrt 2 - 8 \epsilon &\leq 8 \Pr[A B = X \oplus Y] - 4 \\
&= \bra \psi \Big( \sum_{a, b \in \{0,1\}} (-1)^{a b} \RAa{a} \otimes \RBa{b} \Big) \ket \psi \\
&= 2 \sqrt 2 - \sqrt 2 \bra \psi (M_0^2 + M_1^2) \ket \psi
 \enspace .
\end{split}\end{equation*}
For $\epsilon = 0$, this means that $\ket \psi$ must lie in the intersection of the kernels of $M_0$ and~$M_1$.  The four eigenvalues of $M_0$ are $\pm \cos \theta \pm \frac{1}{\sqrt 2}$.  For the kernel to be nonempty, it must be that $\theta = \frac\pi4$.  A symmetrical argument implies that $\theta' = \frac\pi4$.  For small $\epsilon > 0$, $\ket \psi$ must lie close to small-eigenvalue subspaces of both $M_0$ and $M_1$, implying that $\theta$ and~$\theta'$ are close to $\frac\pi4$.  Thus the measurement operators are rigidly determined.  

For $\theta = \theta' = \frac\pi4$, the kernel of $\sqrt 2 ((H G) \otimes I) M_0 ((G^\dagger H) \otimes I) = \sigma_z \otimes I - I \otimes \sigma_z$ is spanned by the vectors $\ket{00}$ and $\ket{11}$.  The kernel of $\sqrt 2 ((H G) \otimes I) M_1 ((G^\dagger H) \otimes I) = \sigma_x \otimes I - I \otimes \sigma_x$ is spanned by the vectors $\ket{+} \otimes \ket{+} = \frac12 (\ket{00} + \ket{01} + \ket{10} + \ket{11})$ and $\ket{-} \otimes \ket{-} = \frac12 (\ket{00} - \ket{01} - \ket{10} + \ket{11})$.  For the $\ket{01}$ and $\ket{10}$ terms to cancel out, a linear combination of these vectors must have equal coefficients.  The intersection between the two kernels is therefore spanned by $\ket{00} + \ket{11}$.  Thus the state $\ket \psi$ is rigidly determined.  

The above argument conveys much of the intuition for the CHSH rigidity theorem.  The case $\epsilon > 0$ can be handled by maintaining suitable approximations.  However, we have not explained the derivation of the operators $M_0$ and~$M_1$, chosen to satisfy $\sum_{a, b \in \{0,1\}} (-1)^{a b} \RAa{a} \otimes \RBa{b} = 2 \sqrt 2 I \otimes I - \sqrt 2 (M_0^2 + M_1^2)$.  In general, for a game in which Eve draws her questions from the distribution $p(a,b)$ and accepts if $x \oplus y = V(a,b)$, let $\Theta = \sum_{a,b} p(a,b) (-1)^{V(a,b)} \ketbra a b$ and $\hat \Theta = \smatrx{0&\Theta\\\Theta^\dagger&0}$.  Let $\omega^*$ be the optimal success probability.  By the Tsirelson semi-definite program~\cite{CleveSlofstraUngerUpadhyay06parallelXOR}, the optimal bias is $2 \omega^* - 1 = \frac12 \max_{\Gamma \succeq 0, \Gamma \circ I = I} \langle \hat \Theta, \Gamma \rangle = \frac12 \min_{\Delta = \Delta \circ I \succeq \hat \Theta} \Tr \Delta$.  $\Gamma$ is the Gram matrix of the vectors $\RAa{a} \ket \psi$ and $\RBa{b} \ket \psi$.  Letting $\Delta^*$ achieve the second optimum, we have $\frac12 \langle \hat \Theta, \Gamma \rangle = (2\omega^*-1) - \frac12 \langle \Delta^* - \hat \Theta, \Gamma \rangle$.  For the CHSH game, $\Delta^* = \frac{1}{2 \sqrt 2} \identity$, and the matrices $M_0, M_1$ correspond to eigenvectors of $\Delta^* - \hat \Theta$.

\section{Rigidity for sequential CHSH games}

For sequentially repeated CHSH games, our goal is to locate in $\H_A$ and $\H_B$ not one but many qubits, in tensor product, such that the devices' actual strategy is close to the ideal strategy that measures these qubits one at a time in sequence.  In this section, we will introduce the notation and put together the theorems needed to prove rigidity for sequential CHSH games.

\subsection{Notation}

To make our claims precise, we begin with some notation for CHSH games played in sequence, one following the next, with no communication between games.  

A strategy $\S$ for Alice and Bob to play $n$ sequential CHSH games consists of the devices' Hilbert spaces, initial shared state and the reflections they use to play each game.  

\begin{description}[leftmargin=.1cm, style=sameline]
\addtolength{\itemsep}{-0.25\baselineskip}
\item[Initial state:]
Let $\H_A$ and $\H_B$ be Alice and Bob's respective Hilbert spaces, and $\H_C$ any external space.  Let $\ket{\psione} \in \H_A \otimes \H_B \otimes \H_C$ be the devices' initial shared state.  

\item[Transcripts:] 
Denote questions asked to Alice by $a_1, \ldots, a_n$, questions asked to Bob by $b_1, \ldots, b_n$, and possible answers by $x_1, \ldots, x_n$ and $y_1, \ldots, y_n$, respectively.  Write $\hA{j} = (a_1, x_1, \ldots, a_j, x_j)$, $\hB{j} = (b_1, y_1, \ldots, b_j, y_j)$ and $\h{j} = (\hA{j}, \hB{j})$, a full transcript for games $1$ through~$j$.  Write $\h{j,k}$ and $\hX{j,k}$ for the full or partial transcripts for games~$j$ through~$k$, inclusive.  

\item[Reflections:] 
In game $j$, for questions $a_j$ and $b_j$, let~$\RAjah{j}{a_j}{\hA{j-1}}$ and $\RBjah{j}{b_j}{\hB{j-1}}$ be the reflections specifying Alice and Bob's respective strategies for game~$j$, depending on the previous games' transcript.  Define projections~$\PAjh{j}{\hA{j}} = \tfrac12 (\identity + (-1)^{x_j} \RAjah{j}{a_j}{\hA{j-1}})$ and $\PBjh{j}{\hB{j}} = \tfrac12 (\identity + (-1)^{y_j} \RBjah{j}{b_j}{\hB{j-1}})$.  For $\device \in \{A, B\}$ and $j \leq k$, let $\PXjh{j,k}{\hX{k}} = \PXjh{k}{\hX{k}} \cdots$ $\PXjh{j+1}{\hX{j+1}} \PXjh{j}{\hX{j}}$.  Let $\PABjh{j, k}{\h{k}} = \PAjh{j, k}{\hA{k}} \otimes \PBjh{j, k}{\hB{k}}$.  

\item[Super-operators:]
For $j < k$ and partial transcript~$\h{j}$, define super-operators~$\EAjh{k}{\hA{j}}$ and~$\EBjh{k}{\hB{j}}$ by 
\begin{equation}\begin{split}
\EXjh{k}{\hX{j}} &(\ketbra{\hX{j+1,k-1}}{\hX{j+1,k-1}} \otimes \rho) \\
&= \frac12 \sum_{\alpha_k, \chi_k} \ketbra{\hX{j+1,k}}{\hX{j+1,k}} \otimes \PXjh{k}{\hX{k}} \rho \PXjh{k}{\hX{k}} 
 \enspace ,
\end{split}\end{equation}
where $\hX{k} = (\hX{k-1}, \alpha_k, \chi_k)$.  These super-operators capture the effects of Alice and Bob playing game~$k$, where games~$j+1$ to~$k-1$ of the transcript are stored in a separate register.  For $\ell \geq k$, let $\EXjh{k, \ell}{\hX{j}} = \EXjh{\ell}{\hX{j}} \cdots \EXjh{k+1}{\hX{j}} \EXjh{k}{\hX{j}}$ and $\EABjh{k, \ell}{\h{j}} = \EAjh{k, \ell}{\hA{j}} \otimes \EBjh{k, \ell}{\hB{j}}$.  

Let $\rhoone = \ketbra{\psione}{\psione}$, $\rhoj{j}(\h{j-1}) = \PABjh{1,j-1}{\h{j-1}} \rhoone \PABjh{1,j-1}{\h{j-1}}^\dagger$ $/ \Tr (\PABjh{1,j-1}{\h{j-1}} \rhoone)$ be the state at the beginning of game~$j$ conditioned on $\h{j-1}$, and $\rhoj{j} = \EABj{1,j-1}(\rhoone) = \frac{1}{4^{j-1}} \sum_{\h{j-1}} \ketbra{\h{j-1}}{\h{j-1}} \otimes \PABjh{1,j-1}{\h{j-1}} \rhoone \PABjh{1,j-1}{\h{j-1}}^\dagger\!$.  Compare to Eq.~\eqnref{e:examplecombinedtranscriptstatedensitymatrix} in the main text.  
\end{description}

For fixed Hilbert spaces, a strategy $\S$ can be identified with the tuple $(\rhoone, \{ \EAj{j} \}, \{ \EBj{j} \})$.  When considering multiple strategies, we will decorate this notation to indicate the corresponding strategy, e.g., $\tilde \S = (\rhodecone{\tilde}, \{ \EAdecj{\tilde}{j} \}, \{ \EBdecj{\tilde}{j} \})$.  

\smallskip

Call a strategy for a single CHSH game $\epsilon$-structured if the probability of winning is at least $\omega^* - \epsilon/8$.  In our theorem, we will assume that most games the devices play are $\epsilon$-structured, in the following sense: 

\begin{definition}[Structured strategy] \label{t:structuredstrategydef}
A sequential CHSH game strategy~$\S$ is \emph{$(\delta, \epsilon)$-structured} if for every~$j$, there is at least a $1-\delta$ probability over transcripts $\h{j-1}$ that game~$(j, \h{j-1})$ is $\epsilon$-structured.  $\S$ is \emph{$\epsilon$-structured} if it is $(\epsilon, \epsilon)$-structured.  
\end{definition}

We aim to show that the devices play nearly ideally: 

\begin{definition}[Ideal strategy] \label{t:idealstrategy}
A strategy~$\S$ for $n$ sequential CHSH games is an \emph{ideal strategy} if there exist isometries $\UidealX: \H_\device \hookrightarrow (\C^2)^{\otimes n} \otimes \H_\device'$ and a state $\ket{\psione'} \in \H_A' \otimes \H_B' \otimes \H_C$ such that for every $j$ and $\h{j-1}$, 
\begin{equation}\begin{split} \label{e:idealstrategy}
\UidealA \otimes \UidealB \ket{\psione} &= \ket{\varphi}^{\otimes n} \otimes \ket{\psione'} \\
\RXjah{j}{\alpha}{\hX{j-1}} &= \UidealX{}^\dagger (\RXa{\alpha})_j \UidealX 
 \enspace ,
\end{split}\end{equation}
where $(\RXa{\alpha})_j$ denotes the ideal reflection operator~$\RXa{\alpha}$ from \figref{f:chsh} in the main text acting on the $j$th qubit.  
\end{definition}

In order to compare strategies, define a notion of simulation: 

\begin{definition}[Strategy simulation] \label{t:simulationdef}
Let $\S$ and $\tilde \S$ be two strategies for playing $n$ sequential CHSH games.  For $\epsilon \geq 0$, we say that strategy $\tilde \S$ \emph{$\epsilon$-simulates} strategy~$\S$ if they both use the same Hilbert spaces and for all~$j$, 
\begin{equation}
\max_{\device \in \{A, B\}} \trnorm{\EXj{1,j}(\rhoone) -  \EXdecj{\tilde}{1,j}(\rhodecone{\tilde})} \leq \epsilon
 \enspace .
\end{equation}
Say that $\tilde \S$ \emph{weakly $\epsilon$-simulates}~$\S$ if only the weaker inequality $\trnorm{ \EABj{1,j}(\rhoone) - \EABdecj{\tilde}{1,j}(\rhodecone{\tilde}) } \leq 2 \epsilon$ holds.  
\end{definition}

It is also convenient to allow a basis change by local unitaries or local isometries: 

\begin{definition} \label{t:isometricextensiondef}
A strategy $\tilde \S$ is an \emph{isometric extension} of $\S$ if there exist isometries $\XX: \H_\device \hookrightarrow \tilde \H_\device$ such that $\ket{\tilde \psione} = \XA \otimes \XB \ket{\psione}$ and $\XX \RXjah{j}{\alpha}{\hX{j-1}} = \RXdecjah{\tilde}{j}{\alpha}{\hX{j-1}} \XX$ always.  (Thus $\XX \EXj{j} = \EXdecj{\tilde}{j} \XX$.)  
\end{definition}

\subsection{Main rigidity theorem and proof outline}

Our main theorem states that a structured strategy can be closely simulated by an ideal strategy: 

\def\kappaEPR{\kappa}	% overriding the previous definition

\begin{theorem}[Rigidity theorem for sequential CHSH games] \label{t:sequentialCHSHgames}
There exists a constant~$\kappaEPR$ such that for any $\epsilon$-structured strategy $\S$ for $n$ sequential CHSH games, there exists an ideal strategy $\hat \S$ that $\kappaEPR n^{\kappaEPR} \epsilon^{1/\kappaEPR}$-simulates an isometric extension of~$\S$.  
\end{theorem}

\begin{figure*}
\centering
\def\qubitsballscale{.28}
\subfigure[General strategy]{\includegraphics[scale=\qubitsballscale]{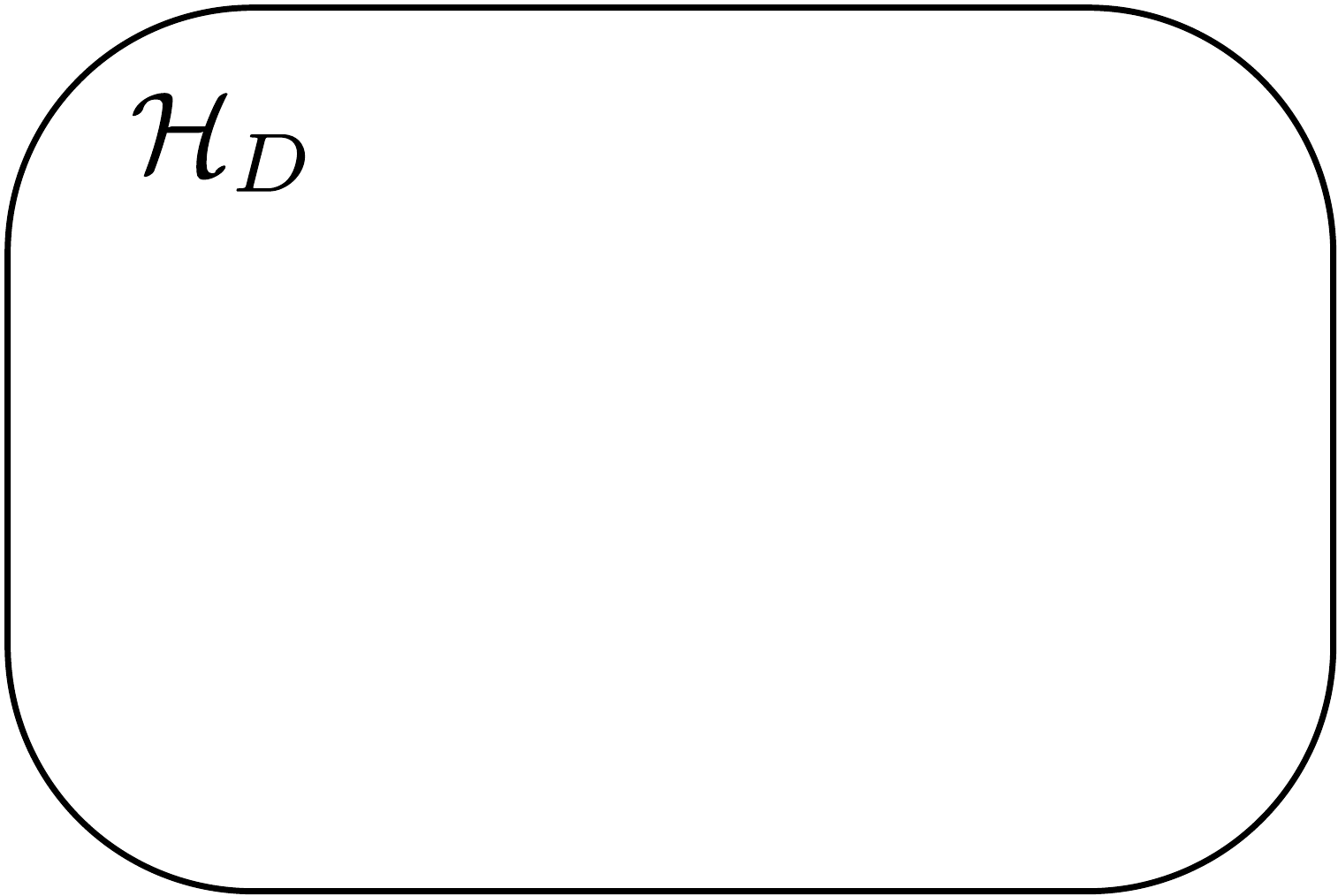}}
\subfigure[Single-qubit ideal strategy]{\includegraphics[scale=\qubitsballscale]{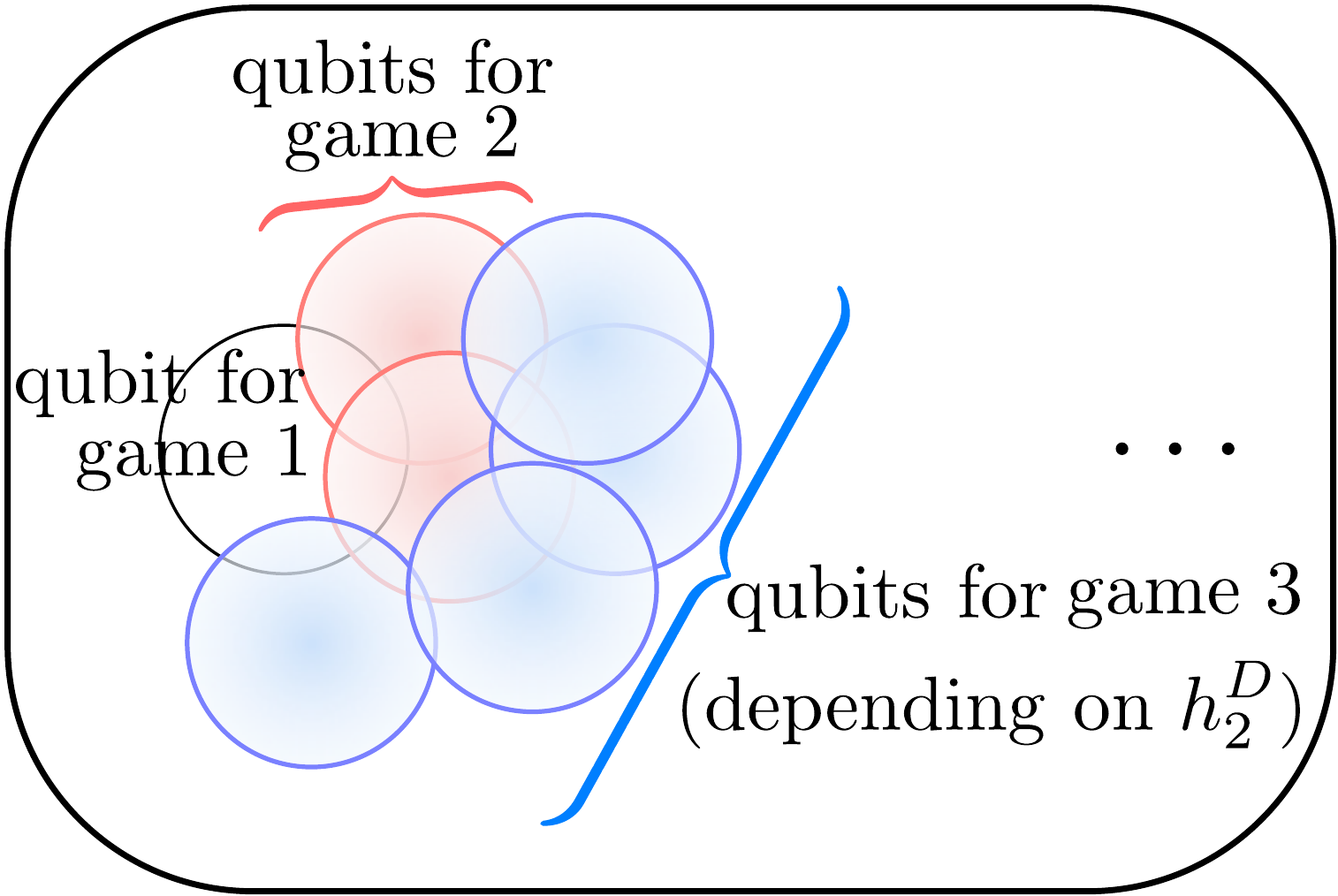}}
\subfigure[Multi-qubit ideal strategy]{\includegraphics[scale=\qubitsballscale]{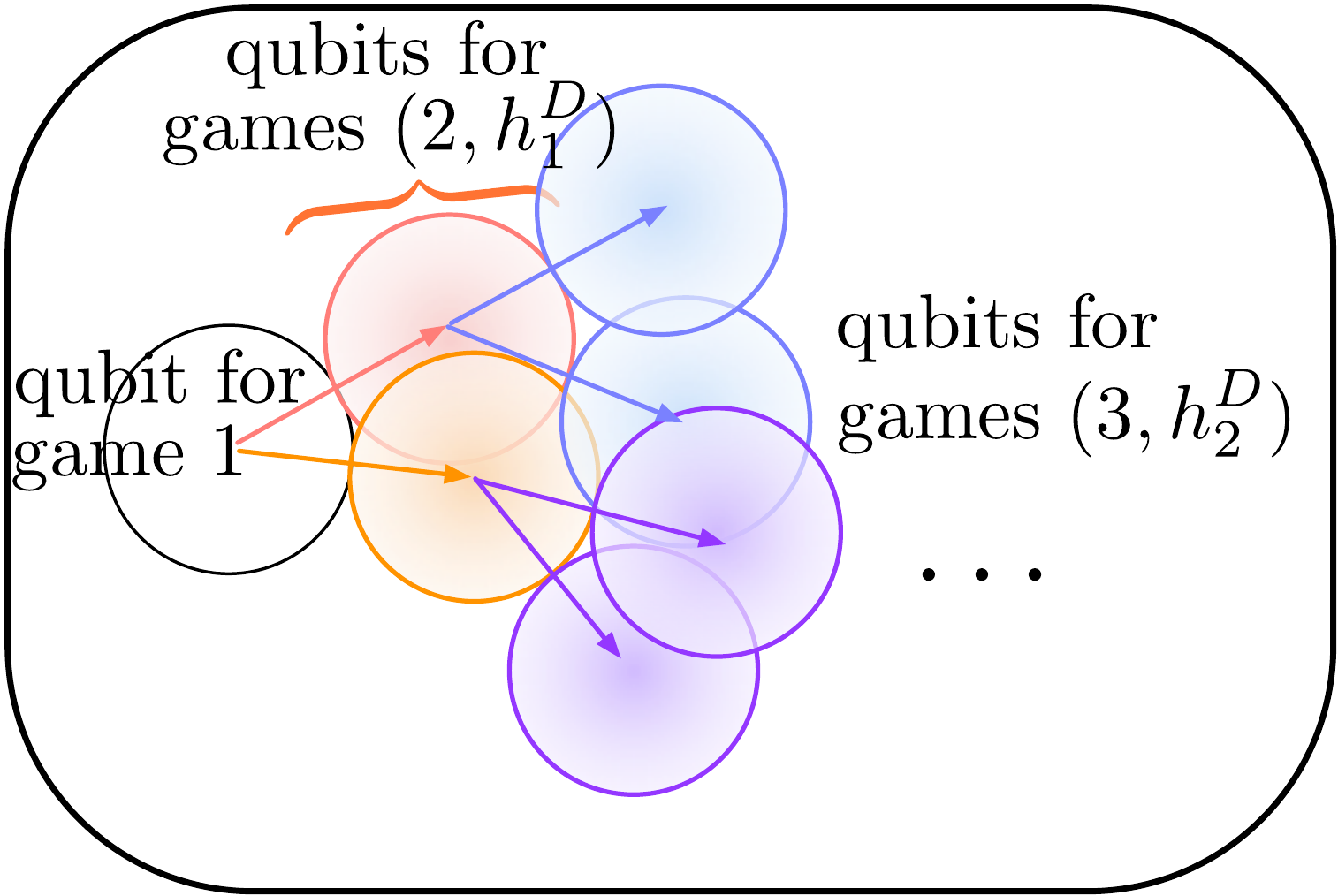}}
\subfigure[Ideal strategy]{\includegraphics[scale=\qubitsballscale]{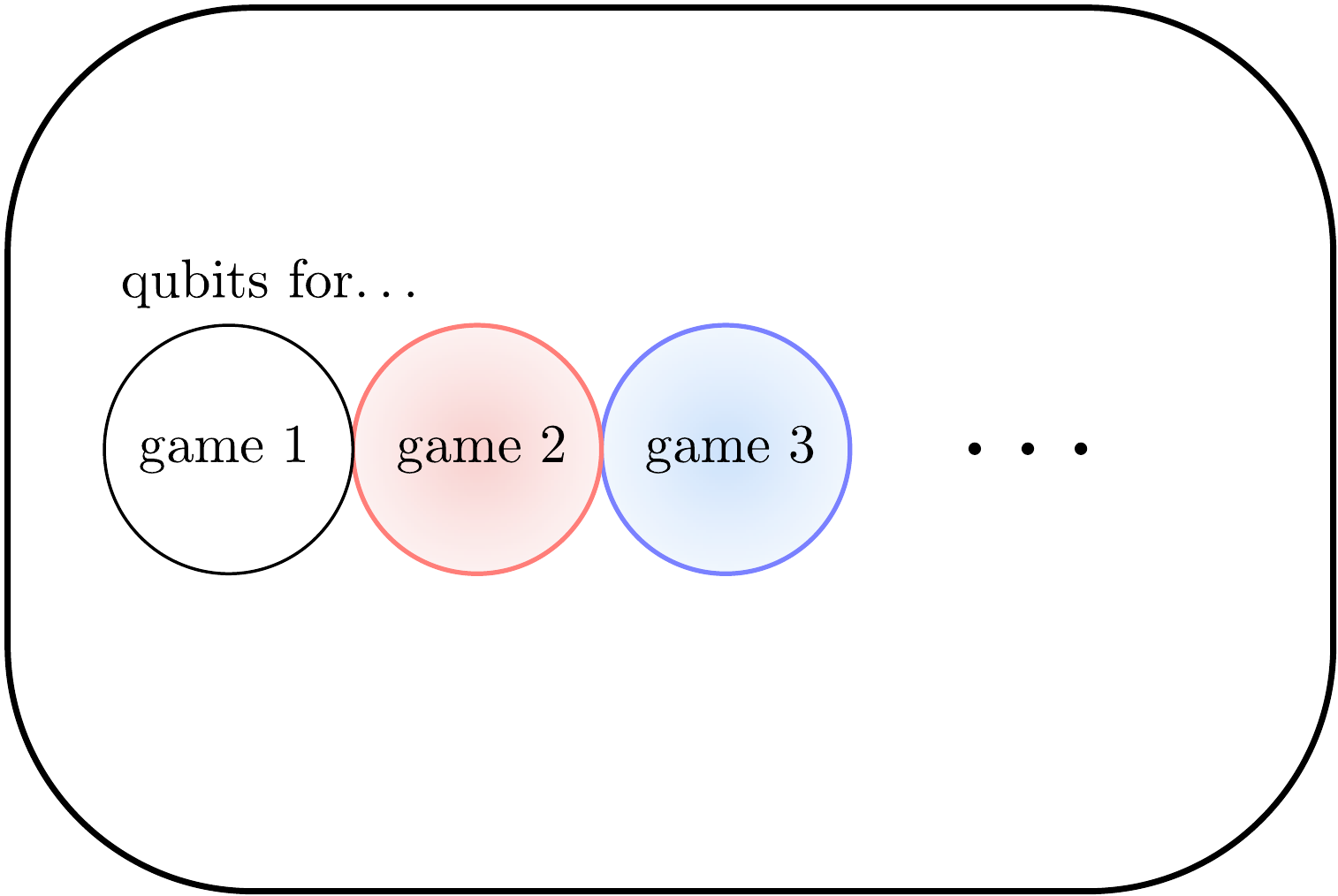}}
\caption{Proof outline for \thmref{t:sequentialCHSHgames}.  
{\bf a}, Initially, each device $D \in \{A, B\}$ can play arbitrarily, measuring in game~$j$ one of two reflections $\RXjah{j}{\alpha}{\hX{j-1}}$, $\alpha \in \{0,1\}$, that depend on the local transcript $\hX{j-1}$ for the previous games.  No structure is given for the Hilbert space $\H_D$.  
{\bf b}, We first show that $D$'s strategy is close to a ``single-qubit ideal strategy," in which for every game it measures some qubit using the ideal CHSH game strategy, but the qubit locations can be arbitrary.  Here, the qubits are illustrated schematically as balls, and the overlaps indicate that they need not be in tensor product.  
{\bf c}, We then construct a nearby ``multi-qubit ideal strategy," in which the qubits used in each game must lie in tensor product with the qubits from previous games, but can overlap qubits used along other transcripts.  
{\bf d}, Finally, we argue that the qubit locations cannot depend significantly on the transcript, and therefore that the original strategy is well-approximated by an ideal strategy that measures a fixed set of $n$ qubits in sequence.  
(Note that these visualizations, representing qubits as balls, are inherently imprecise.  A qubit's location in a Hilbert space is given not by a ball, but by the two anti-commuting reflection operators $\sigma_x$ and $\sigma_z$.)  
} \label{f:proofsketch}
\end{figure*}

The first step of the proof of \thmref{t:sequentialCHSHgames} is to replace the structured strategy $\S$ with one in which the devices play every game using the ideal CHSH game operators on some qubit, up to a local change in basis.  See \figref{f:proofsketch}.  

\begin{definition}[Single-qubit ideal strategy] \label{t:singlequbitidealstrategy}
A strategy~$\S$ is a \emph{single-qubit ideal strategy} if there exist unitaries $\UsingleXjh{j}{\hX{j-1}} : \H_\device \overset{\cong}{\rightarrow} \C^2 \otimes \H_\device'$ such that always 
\begin{equation} \label{e:singlequbitidealstrategy}
\RXjah{j}{\alpha}{\hX{j-1}} = \UsingleXjh{j}{\hX{j-1}}^\dagger (\RXa{\alpha} \otimes \identity) \UsingleXjh{j}{\hX{j-1}} 
 \enspace .
\end{equation}
That is, each device's reflections for game $(j, \hX{j-1})$ are equivalent up to local unitaries to the ideal CHSH game reflections, although the qubits used need not be in tensor product.  
\end{definition}

\begin{theorem} \label{t:singlequbitidealstrategysimulation}
There exists a constant~$\kappa$ such that if 
$\S$ is an $\epsilon$-structured strategy for~$n$ sequential CHSH games, then there is a single-qubit ideal strategy~$\tilde \S$ that $\kappa n^\kappa \epsilon^{1/\kappa}$-simulates an isometric extension of~$\S$.  
\end{theorem}

\begin{proof}[Proof sketch]
As explained in the main text, let $\EXdecj{\tilde}{j}$ be the super-operator that replaces the device's measurement operators with the ideal operators promised by the CHSH rigidity theorem.  Then $\tilde \S = (\rhoone, \{ \EAdecj{\tilde}{j} \}, \{ \EBdecj{\tilde}{j} \})$.  If $\Pr[\text{game~$j$ is $\epsilon$-structured}] \geq 1-\delta$, then $\trnorm{\EABj{j}(\rhoj{j}) - \EABdecj{\tilde}{j}(\rhoj{j})} \leq 2 \delta + O(\sqrt \epsilon)$.  (This expression combines bounds on the probability of the bad event and the $O(\sqrt \epsilon)$ error from the good event.)  

To show our goal, that $\EABj{1,n}(\rhoone) \approx \EABdecj{\tilde}{1,n}(\rhoone)$ in trace distance, use a hybrid argument that works backwards from game~$n$ to game~$1$ fixing each game's measurement operators one at a time.  The error introduced from fixing a game~$j$, by moving from $\EABj{j}(\rhoj{j})$ to $\EABdecj{\tilde}{j}(\rhoj{j})$, does not increase in later games because applying a super-operator cannot increase the trace distance.  Mathematically, this hybrid argument is simply a triangle inequality using the expansion 
\begin{equation*}
\EABj{1,n}(\rhoone) - \EABdecj{\tilde}{1,n}(\rhoone) = 
\sum_{j \in [n]} \EABdecj{\tilde}{j+1,n} \big(\EABj{j}(\rhoj{j}) - \EABdecj{\tilde}{j}(\rhoj{j})\big)
 \enspace .  
\end{equation*}

Therefore $\tilde \S$ weakly simulates $\S$; the extension to simulation is more technical, as described in the main text.  
\end{proof}

Next, we find a nearby strategy in which the qubits for successive games are in tensor product.  

\begin{definition}[Multi-qubit ideal strategy] \label{t:multiqubitidealstrategy}
A strategy~$\S$ is a \emph{multi-qubit ideal strategy} if there is a unitary isomorphism $\XmultiX : \H_\device \overset{\cong}{\rightarrow} (\C^2)^{\otimes n} \otimes \H_\device'$ under which for unitaries $\UmultiXjh{j}{\hX{j-1}} \in \L((\C^2)^{\otimes (n-j+1)} \otimes \H_\device')$ such that 
\begin{equation}\begin{split} \label{e:multiqubitidealstrategy}
\RXjah{j}{\alpha}{\hX{j-1}} = \qquad\qquad\qquad\qquad\qquad\qquad\quad
\XmultiX{}^\dagger
\UmultiXj{1}{}^\dagger \ldots \\ \quad \big(\identity_{\C^2}^{\otimes (j-1)} 
\otimes \UmultiXjh{j}{\hX{j-1}}^\dagger\big) (\RXa{\alpha})_j \big(\identity_{\C^2}^{\otimes (j-1)} \otimes \UmultiXjh{j}{\hX{j-1}}\big) \\ \ldots 
\UmultiXj{1}
\XmultiX
 \enspace .
\end{split}\end{equation}
That is, $\S$ is a single-qubit ideal strategy in which the qubits used in each game must lie in tensor product with the qubits from previous games.  
\end{definition}

\begin{theorem} \label{t:multiqubitidealstrategysimulation}
There exists a constant~$\kappa$ such that if $\tilde \S$ is an $\epsilon$-structured single-qubit ideal strategy for~$n$ sequential CHSH games, then there is a multi-qubit ideal strategy~$\bar \S$ that $\kappa n^\kappa \epsilon^{1/\kappa}$-simulates an isometric extension of~$\tilde \S$.  
\end{theorem}

\begin{proof}[Proof sketch]
The tensor-product structure is constructed beginning with a trivial transformation on~$\tilde \S$: to each device, add~$n$ ancilla qubits each in state~$\ket 0$.  Next, after a qubit has been measured, say as~$\ket{\alpha_j}$ in game~$j$, swap it with the $j$th ancilla qubit, then rotate this fresh qubit from $\ket 0$ to $\ket{\alpha_j}$ and continue playing games $j+1, \ldots, n$.  This defines a unitary change of basis that places the outcomes for games $1$ to~$j$ in the first $j$ ancilla qubits, and leaves the state in the original Hilbert space unchanged.  Since qubits are set aside after being measured, the qubits for later games are automatically in tensor product with those for earlier games; the resulting strategy $\bar \S$ is multi-qubit ideal.  At the end of the $n$ games, swap back the ancilla qubits and undo their rotations, using the transcript.  

The key to showing that $\bar \S$ is close to $\tilde \S$ is the fact that operations on one half of an EPR state can equivalently be performed on the other half, since for any $2 \times 2$ matrix~$M$, $(M \otimes I)(\ket{00} + \ket{11}) = (I \otimes M^T)(\ket{00} + \ket{11})$.  This means that the outcome of an $\epsilon$-structured CHSH game would be nearly unchanged if Bob were hypothetically to perform Alice's measurement before his own.  By moving Alice's measurement operators for games $j+1$ to $n$ over to Bob's side, we see that they cannot significantly affect the qubit $\ket{\alpha_j}$ from game~$j$ on her side.  Therefore, undoing the original change of basis restores the ancilla qubits nearly to their initial state $\ket{0^n}$, and $\tilde \S \approx \bar \S$.    

\def\AmeasBBdecj #1#2{#1{\F}^{AB}_{#2}}

Formally, define a unitary super-operator ${\cal V}_j$ that rotates the $j$th ancilla qubit to $\ket{\alpha_j}$, depending on Alice's local transcript~$\hA{j}$.  Define a unitary super-operator ${\cal T}_j$ to apply ${\cal V}_j$ and swap the $j$th ancilla qubit with the qubit Alice uses in game~$j$ (depending on~$\hA{j-1}$).  Alice's multi-qubit ideal strategy is given by 
\begin{equation}
\EAdecj{\bar}{j} = {\cal T}_{1,j-1}^{-1} (\identity_{\C^{2^n}} \otimes \EAdecj{\tilde}{j}) {\cal T}_{1,j-1}
 \enspace .
\end{equation}
We aim to show that the strategy given by $\rhoone$, $\{ \EAdecj{\bar}{j} \}$ and $\{ \EBdecj{\tilde}{j} \}$ is close to~$\tilde \S$ up to the fixed isometry that prepends $\ketbra{0^n}{0^n}$ to the state, i.e., that $\ketbra{0^n}{0^n} \otimes \EABdecj{\tilde}{1,n}(\rhoone) \approx \EAdecj{\bar}{1,n}\big( \ketbra{0^n}{0^n} \otimes \EBdecj{\tilde}{1,n} (\rhoone) \big)$.  Define a super-operator $\AmeasBBdecj{\tilde}{j}$, in which Alice's measurements are made on \emph{Bob's} Hilbert space~$\H_B$, on the qubit determined by Bob's local transcript $\hB{j-1}$.  Since most games are $\epsilon$-structured, by the CHSH rigidity theorem, $\AmeasBBdecj{\tilde}{j+1,k}(\rhodecj{\tilde}{j+1}) \approx \EABdecj{\tilde}{j+1,k}(\rhodecj{\tilde}{j+1}) = \rhodecj{\tilde}{k+1}$ for any $j \leq k$.  Since $\AmeasBBdecj{\tilde}{j+1,k}$ acts on $\H_B$, it does not affect Alice's qubit $\ket{\alpha_j}$ from game~$j$ at all, and so this qubit must stay near $\ket{\alpha_j}$ in $\rhodecj{\tilde}{k+1}$ as well, i.e., the trace of the reduced density matrix against the projection $\ketbra{\alpha_j}{\alpha_j}$ stays close to one.  As this holds for every~$j$, ${\cal T}_{1,n}^{-1}$ indeed returns the ancillas almost to their initial state~$\ket{0^n}$.  

In more detail, let~$X_j$ be the operator that projects onto Alice's $j$th ancilla qubit and the qubit she uses in the $j$th game being $\ket{0} \otimes \ket{\alpha_j}$.  By definition, $\Tr (X_j \, \rhodecj{\tilde}{j+1}) = 1$.  By the Gentle Measurement Lemma~\cite{Winter99coding, OgawaNagaoka07channelcoding}, it suffices to show that $\Tr (X_j \, \rhodecj{\tilde}{k+1}) = \Tr X_j \EABdecj{\tilde}{j+1,k}(\rhodecj{\tilde}{j+1}) \approx 1$.  This is not obvious; since the operators for games~$j+1$ to $k$ do not act in tensor product, they can disturb the qubit measured in game~$j$.  However, since a super-operator on~$\H_B$ cannot affect the expectation of an operator supported on~$\H_A$, we find 
\begin{equation*}\begin{split}
\Tr (X_j \, \rhodecj{\tilde}{k+1}) &= \Tr X_j \EABdecj{\tilde}{j+1,k}(\rhodecj{\tilde}{j+1}) \\ &\approx \Tr X_j \AmeasBBdecj{\tilde}{j+1,k}(\rhodecj{\tilde}{j+1}) \\ &= \Tr(X_j  \, \rhodecj{\tilde}{j+1}) \\ &= 1
 \enspace .
\end{split}\end{equation*}

A symmetrical argument adjusts Bob's super-operators $\{ \EBdecj{\tilde}{j} \}$ to $\{ \EBdecj{\bar}{j} \}$, implying that $\bar \S$ weakly simulates $\tilde \S$.  
\end{proof}

The last major step in the proof of \thmref{t:sequentialCHSHgames} is to argue that the qubit locations cannot depend significantly on the local transcripts, and therefore simulate a multi-qubit ideal strategy with an ideal strategy.  

\begin{theorem} \label{t:globalgluing}
There exists a constant~$\kappa$ such that if~$\bar \S$~is an $\epsilon$-structured multi-qubit ideal strategy for~$n$ sequential~CHSH games, then there is an ideal strategy that $\kappa n^\kappa \epsilon^{1/\kappa}$-simulates~$\bar \S$.  
\end{theorem}

\begin{proof}[Proof sketch]
\def\UmultiXdecj #1#2{#1{M}^{\device}_{#2}}
\def\UmultiXdecjh #1#2#3{#1{M}^{\device}_{#2}({#3})}
\def\XmultiXdec #1{#1{\mathcal{Y}}^{\device}}
Fix a transcript $\hdec{\hat}{n}$, chosen at random from the distribution of transcripts for $\bar \S$.  Define a strategy $\hat \S$ to use the same initial state $\rhodecone{\bar}$ as $\bar \S$ and the qubits specified by $\hdec{\hat}{n}$ in~$\bar \S$, i.e., $\RXdecjah{\hat}{j}{\alpha}{\hX{j-1}} = \RXdecjah{\hat}{j}{\alpha}{\hXdec{\hat}{j-1}}$ independent of~$\hX{j-1}$.  Thus, as required in \defref{t:idealstrategy}, $\RXdecjah{\hat}{j}{\alpha}{\hX{j-1}} = \UidealXdec{\hat}{}^\dagger (\RXa{\alpha})_j \UidealXdec{\hat}$ for 
\begin{equation*}
\UidealXdec{\hat} = \big(\identity_{(\C^2)^{\otimes (n-1)}} \otimes \UmultiXdecjh{\bar}{n}{\hX{n-1}}\big) \ldots \UmultiXdecj{\bar}{1} \XmultiXdec{\bar}
 \enspace .
\end{equation*}

We argue that $\hat \S$ closely approximates $\bar \S$, provided that~$\hdec{\hat}{n}$ satisfies: for every~$j$, conditioned on the partial transcript~$\hdec{\hat}{j-1}$, 
(a) game~$j$ is $\epsilon$-structured, and 
(b)~there is a high probability that every subsequent game is $\epsilon$-structured.  By Markov inequalities, most transcripts satisfy these conditions.  

\def\rhodecjh #1#2#3{#1{\rho}_{#2}({#3})}	% cannot hide the j argument
\def\AmeasBBdecjh #1#2#3{#1{\F}^{AB \vert \smash{#3}}_{#2}}	% Alice's measurement on Bob's qubit followed by Bob's measurement
\def\EABdecjh #1#2#3{#1{\E}^{AB \vert \smash{#3}}_{#2}}

We connect $\bar \S$ to $\hat \S$ by an argument that one game at a time switches play to locate qubits according to $\hdec{\hat}{n}$.  The intermediate steps relate strategies in which the devices locate their qubits using a hybrid $(\hdec{\hat}{j}, \h{j+1,n})$ of $\hdec{\hat}{n}$ and the actual transcript $\h{n}$.  

Consider a partial transcript $\h{j}$ that differs from $\hdec{\hat}{j}$ only in the $j$th game, say on Alice's side.  By~(a) and the CHSH rigidity theorem, Alice's $j$th qubit is collapsed and nearly in tensor product with the rest of the state.  Therefore, there exists a unitary~$\AAunitaryAj{j}$ acting on this qubit such that 
\begin{equation} \label{e:proofidealocalgluingassumption}
\rhodecjh{\bar}{j+1}{\h{j}} \approx \AAunitaryAj{j} \rhodecjh{\bar}{j+1}{\hdec{\hat}{j}} \AAunitaryAj{j}{}^\dagger
 \enspace ,
\end{equation}
up to error $O(\sqrt \epsilon)$.  Since applying a super-operator cannot increase trace distance and on Bob's side $\hB{j} = \hBdec{\hat}{j}$, therefore 
\begin{equation*}
\AmeasBBdecjh{\bar}{j+1,n}{\hB{j}}\big(\rhodecjh{\bar}{j+1}{\h{j}}\big) \approx \AAunitaryAj{j} \AmeasBBdecjh{\bar}{j+1,n}{\hBdec{\hat}{j}}\big(\rhodecjh{\bar}{j+1}{\hdec{\hat}{j}}\big) \AAunitaryAj{j}{}^\dagger
 \enspace .
\end{equation*}
Here, $\AmeasBBdecjh{\bar}{j+1,n}{\hB{j}}$ is the same super-operator used in the multi-qubit ideal strategy simulation step---that plays Alice's games on Bob's qubits---except conditioned on the local transcript~$\hB{j}$.  By condition (b), these super-operators can be pulled back to Alice's side, to give 
\begin{equation*}
\EABdecjh{\bar}{j+1,n}{\h{j}}\big(\rhodecjh{\bar}{j+1}{\h{j}}\big) \approx \AAunitaryAj{j} \EABdecjh{\bar}{j+1,n}{\hdec{\hat}{j}}\big(\rhodecjh{\bar}{j+1}{\hdec{\hat}{j}}\big) \AAunitaryAj{j}{}^\dagger
 \enspace .
\end{equation*}
Note that this approximation does not follow immediately from Eq.~\eqnref{e:proofidealocalgluingassumption}, because Alice's super-operators conditioned on $\hA{j}$ can be very different from her super-operators conditioned on~$\hAdec{\hat}{j}$.  

By fixing the coordinates one at a time in this way, we find that for a typical transcript~$\h{n}$, $\rhodecjh{\bar}{n+1}{\h{n}} \approx V^{AB}_{1,n} \rhodecjh{\bar}{n+1}{\hdec{\hat}{n}} V^{AB}_{1,n}{}^\dagger$, and we conclude that $\EABdecj{\bar}{1,n}(\rhoone) \approx \EABdecj{\hat}{1,n}(\rhoone)$.  

Since $\EABdecj{\hat}{1,n}$ measures qubits in tensor product with each other, by using the CHSH rigidity theorem one last time, it is not difficult to show that $\EABdecj{\hat}{1,n}(\rhoone) \approx \EABdecj{\hat}{1,n}(\rhodecone{\hat})$, where $\rhodecone{\hat}$ has~$n$ EPR pairs in the qubit positions determined by~$\hdec{\hat}{n}$.  Thus $\bar \S$ is weakly simulated by the ideal strategy given by $(\rhodecone{\hat}, \{ \EAdecj{\hat}{j} \}, \{ \EAdecj{\hat}{j} \})$.  
\end{proof}

These three steps chain together to yield \thmref{t:sequentialCHSHgames}, via: 

\begin{lemma} \label{t:simulationpreservesstructure}
Let $\S$ be a $(\delta, \epsilon)$-structured strategy for $n$ sequential CHSH games.  If~$\tilde \S$ is a strategy that weakly $\eta$-simulates~$\S$, then $\tilde \S$ is $(\delta + 2 \sqrt{\eta}, \epsilon + 16 \sqrt{\eta})$-structured.  
\end{lemma}

\noindent
This lemma follows immediately from the definitions.

\section{Verified quantum dynamics proof sketches}

\subsection{$X\!Z$-determined states}

\def\operatorset{S}

In the CHSH game, each device has two measurement settings, that in the ideal strategy may be identified with Pauli $\sigma_x$ and~$\sigma_z$ operators.  For carrying out tomography, however, it is generally necessary to be able to measure in the~$\sigma_y$ basis as well.  It is possible to extend the CHSH game to one in which the ideal strategy also uses $\sigma_y$ operators on a shared EPR state.  However, this extended game will not satisfy the rigidity property.  The problem is that a device that consistently measures using $-\sigma_y$ will give indistinguishable statistics from one that uses $+\sigma_y$.  (Switching the sign of $\sigma_y$ is equivalent to taking an entry-wise complex conjugate in the computational basis.)  It is impossible to fix the sign of the $\sigma_y$ operator.  

We therefore instead argue that for certain states, reliable tomography can be accomplished without needing to measure in the $\sigma_y$ basis.  This observation is due to McKague~\cite{McKague10thesis} and was suggested earlier by Magniez et al.~\cite{MagniezMayersMoscaOllivier05selftest}.  McKague shows that for $\ket{\varphi} = \frac{1}{\sqrt 2}(\ket{00} + \ket{11})$, an EPR pair, the states $(I \otimes U) \ket{\varphi}$, for any single-qubit real unitary~$U$, and $\text{CNOT}_{24} \ket{\varphi}_{12} \otimes \ket{\varphi}_{34}$, as well as finite tensor products of these states, are exactly determined by their traces against tensor products of $I$, $\sigma_x$ and~$\sigma_z$ operators.  That is, they are determined by the expectations of observables that can be estimated using measurements in the $\sigma_x$ and~$\sigma_z$ bases.  We call such states ``$X\!Z$-determined."  For our applications, we will need to show a larger class of states to be $X\!Z$-determined.  However, characterizing the full set of $X\!Z$-determined states remains an open problem.   

\begin{definition} \label{t:xzdetermineddef}
For a Hilbert space~$\H$, a set of operators $\operatorset \subseteq \L(\H)$ and $d > 0$, a state $\tau \in \L(\H)$ is \emph{determined by~$\operatorset$ with exponent~$d$} if there exists $c > 0$ such that for all~$\epsilon \geq 0$ and any state $\rho \in \L(\H)$, 
\begin{equation} \label{e:xzdetermineddef}
\max_{P \in \operatorset} \abs{ \Tr P (\rho - \tau) } \leq \epsilon
\quad\Longrightarrow\quad 
\trnorm{\rho - \tau} \leq c \, \epsilon^d
 \enspace .
\end{equation}
For $\H = (\C^2)^{\otimes n}$, a state $\tau$ is \emph{$X\!Z$-determined} if it is determined with some exponent~$d > 0$ by the Pauli operators~$\{I, \sigma_x, \sigma_z\}^{\otimes n}$.  
\end{definition}

Robustness is important for applications, but previous work has considered only $\epsilon = 0$.  By {\Lstroke}ojasiewicz's inequality~\cite[Prop.~2.3.11]{BenedettiRisler90semialgebraic} in algebraic geometry, robustness follows from the $\epsilon = 0$ case: 

\begin{lemma}
For a finite-dimensional Hilbert space~$\H$, a state $\tau \in \L(\H)$ is determined by a finite set $\operatorset \subset \L(\H)$ if and only if for any state $\rho \in \L(\H)$, the implication of Eq.~\eqnref{e:xzdetermineddef} holds at $\epsilon = 0$.  
\end{lemma}

Recall that a \emph{stabilizer state} is an $n$-qubit pure state~$\ket \psi$ for which there exists a set of $2^n$ distinct and pairwise commuting operators $\operatorset \subset \{ \pm P : P \in \{I, \sigma_x, \sigma_y, \sigma_z\}^{\otimes n} \}$, the \emph{stabilizer group}, such that $P \ket \psi = \ket \psi$ for all~$P \in \operatorset$~\cite{NielsenChuang00}.  Any set of $n$ operators that generate the stabilizer group~$\operatorset$ are called stabilizer generators for $\ket \psi$.  We prove: 

\begin{theorem} \label{t:XZdeterminedstabilizerstates}
If $\ket \psi \in (\C^2)^{\otimes n}$ is a stabilizer state that has a set of stabilizer generators in $\{I, \sigma_x, \sigma_z\}^{\otimes n}$, and if $U$ is the tensor product of any~$n$ single-qubit real unitaries, then $U \ket \psi$ is $X\!Z$-determined.  
\end{theorem}

In particular, the resource states needed in our verified quantum dynamics protocol, $\ket 0 \otimes (I \otimes H) \ket{\varphi} \otimes (I \otimes \phasegate) \ket{\varphi} \otimes \mathrm{CNOT}_{2,4} (\ket{\varphi} \otimes \ket{\varphi})$, are $X\!Z$-determined, as are the same states with $G$ applied transversally.  This latter consideration is important because Alice's ideal measurement bases in the CHSH game are rotated by $\pi/4$ from $\sigma_x$ and $\sigma_z$.  

The proof of \thmref{t:XZdeterminedstabilizerstates} begins by showing easily that $\ket 0$ is determined by $\{ \sigma_z \}$.  Then apply three closure properties.  First, closure under tensor product for pure states implies that $\ket{0}^{\otimes n}$ is determined by the $n$ operators $\sigma_z \otimes I^{\otimes (n-1)}, \ldots, I^{\otimes (n-1)} \otimes \sigma_z$.  Second, observe by manipulating implication~\eqnref{e:xzdetermineddef} that if $\tau$ is a state determined by~$\operatorset$, then for any unitary~$U$, $U \tau U^\dagger$ is determined with the same exponent by $\{ U P U^\dagger : P \in \operatorset \}$.  Since an arbitrary stabilizer state can be generated by applying Clifford operators to $\ket{0}^{\otimes n}$, this implies that a stabilizer state is determined by any of its sets of stabilizer generators.  Furthermore, if $\tau$ is determined by~$\operatorset = \{ P_1, \ldots, P_s \}$, then for any invertible $s \times s$ matrix~$V$, $\tau$ is determined with the same exponent by $\{ \sum_{j = 1}^{s} V_{ij} P_j : i = 1, \ldots, s\}$.  This implies, e.g., that any $X\!Z$-determined state is also determined by the operators $\{ I, \frac{1}{\sqrt 2}( \sigma_z \pm \sigma_x ) \}^{\otimes n}$, since the $\{ I, \sigma_x, \sigma_z \}^{\otimes n}$ coefficients are functions of the $\{ I, \frac{1}{\sqrt 2}( \sigma_z \pm \sigma_x ) \}^{\otimes n}$ coefficients.

\subsection{State tomography}

The state tomography protocol begins with $(K-1) m$ CHSH games, where $K$ is chosen randomly.  The multi-game CHSH rigidity theorem implies that at the beginning of the $K$th block of $m$, with high probability Alice and Bob share a state that is close to $m$ shared EPR states, possibly in tensor product with an additional state, and that their separate measurement strategies for the next $m$ games are close to the ideal strategy that uses one EPR state at a time.  For the analysis of the state tomography protocol, we may therefore assume that Alice's strategy is exactly ideal and restrict consideration to these $m$ EPR states.  

Bob, on the other hand, does not play more CHSH games, but instead is given by Eve a random permutation of the $m$ indices, and is asked to permute his qubits and prepare many copies of a particular resource state.  (In the main text, this state is specified as $\ket \psi = \ket 0 \otimes (I \otimes H) \ket{\varphi} \otimes (I \otimes \phasegate) \ket{\varphi} \otimes \mathrm{CNOT}_{2,4} (\ket{\varphi} \otimes \ket{\varphi})$.  The $\ket 0$ portion is for preparing the initial states in a computation and implementing the final measurements, and the other subsystems are for teleporting into each of the gates in a universal gate set, e.g., $(I \otimes G) \ket \varphi$ for teleporting into~$G$.  However, after teleporting into $G$, an $H$ correction may or may not be required.  To maintain the blindness property of the protocol, i.e., to avoid leaking any information about the computation to the separate devices, it is of technical use to have available the resource state $\ket \psi \otimes \ket \varphi$, where the extra EPR state is used for teleporting into the identity gate when an $H$ correction is not needed.)  

Note that Bob's reduced density matrix is maximally mixed, so the probability that he can measure the correct $11$-qubit resource state is only $1/2^{11}$.  However, since the states $\big\{ (P \otimes I) \ket \varphi : P \in \{I, \sigma_x, \sigma_y, \sigma_z \} \big\}$ form an orthonormal basis, so too do the states 
\begin{multline} \label{e:resourcestatebasis}
P^{(0)} \ket 0 \otimes (P^{(1)} \otimes I) \ket \varphi \otimes (P^{(2)} \otimes H) \ket{\varphi} \otimes (P^{(3)} \otimes \phasegate) \ket{\varphi} \\ \otimes 
(P_1^{(4)} \otimes P_3^{(5)} \otimes \mathrm{CNOT}_{2,4}) (\ket{\varphi}_{12} \otimes \ket{\varphi}_{34})
 \enspace ,
\end{multline}
where $P^{(0)} \in \{I, \sigma_x\}$ and the other $P^{(j)}$ vary over $\{I, \sigma_x, \sigma_y, \sigma_z\}$.  Any of the states in Eq.~\eqnref{e:resourcestatebasis} are equally useful resources for computation by teleportation, as Eve can adjust for the $P^{(j)}$ operators in her classical Pauli frame~\cite{Knill05}.  

Therefore define an ideal state tomography protocol as one in which Alice and Bob's initial state consists of $m$ shared EPR states, possibly in tensor product with an additional state; Alice plays honestly $m$ CHSH games, directed by Eve; and Eve sends Bob a random $m$-item permutation and requests that he return the results of measuring consecutive $11$-qubit blocks of permuted qubits in the basis of Eq.~\eqnref{e:resourcestatebasis}.  Eve rejects if the tomography statistics returned by Alice are inconsistent with Bob's reported outcomes.  More precisely, she rejects if the fraction of times Bob reports any particular state differs from $1/2^{11}$ by more than $\sqrt{(\log m) / m}$, i.e., about $\sqrt{\log m}$ standard deviations, or if for any state any of its $\{I,X,Z\}^{\otimes 11}$ Pauli coefficients differ from the corresponding observed estimates by more than $\sqrt{(\log m)/m}$.  We show: 

\begin{theorem} \label{t:statetomographyroughly}
In an ideal state tomography protocol, if Alice plays honestly, then: 
\begin{description}
\item[Completeness:]
If Bob plays honestly, then Eve accepts with high probability.  
\item[Soundness:]
If Eve accepts with high probability, then there is a high probability that, after Bob and before Alice's play, for most of the consecutive $11$-qubit subsystems, Alice's reduced state on the subsystem is close to the state in Eq.~\eqnref{e:resourcestatebasis} that Bob reported to Eve.  
\end{description}
\end{theorem}

In these statements, ``with high probability" means with probability inverse polynomially close to one, i.e., at least $1 - 1/m^c$, where the exact exponents are adjustable.  

The completeness property of the protocol is a trivial application of Hoeffding's inequality; the probability of straying by more than $\sqrt{\log m}$ standard deviations is at most $\exp(- \Omega({\log m})) = m^{-\Omega(1)}$.  The soundness property is also mostly a straightforward tomography argument, using that the states in Eq.~\eqnref{e:resourcestatebasis} are all $X\!Z$-determined.  The main technical complication is that the states of Alice's $11$-qubit blocks need not be in tensor product.  Therefore, her measurement results on different blocks need not be independent.  They can be controlled with a suitable martingale.  

First fix a permutation $\sigma \in S_m$ and a string $x \in (\{0,1\}^{11})^{m/11}$ such that, conditioned on Eve sending Bob $\sigma$ and receiving back~$x$, Eve accepts with high probability.  The remaining randomness consists of Alice's measurement bases and results.  Since Alice's measurements commute, we may assume that she measures her qubits in the permuted order, without changing the measurement statistics.  

For $j = 1, \ldots, m/11$, let $\rho_j$ be Alice's initial reduced state on her $j$th block of $11$ qubits.  Our goal is to control most of the states $\rho_j$.  Let $\sigma_j$ be the state of the same qubits just before she begins to measure them.  The state~$\sigma_j$ is a random variable, but is a deterministic function of the transcript $\hA{11 (j-1)}$ of the earlier CHSH games with Alice.  Conditioned on $\hA{11 (j-1)}$, Alice's measurement results for games $11 (j-1) + 1, \ldots, 11 j$ are distributed according to the Pauli coordinates of~$\sigma_j$.  

For each $b \in \{0,1\}^{11}$, let $\pi_b$ be the corresponding state of Eq.~\eqnref{e:resourcestatebasis}, and let $\tau_b$ be the average of those $\sigma_j$ for which Bob reported $x_j = b$.  Using a martingale argument, we can establish that with high probability, for all $b$, $\pi_b$ and $\tau_b$ have similar $\{I, \sigma_x, \sigma_z\}^{\otimes 11}$ Pauli coordinates.  Since $\pi_b$ is $X\!Z$-determined, this implies that all Pauli coordinates of $\pi_b$ and $\tau_b$ are close.  Since $\pi_b$ is a pure state, i.e., an extremal quantum state, a Markov inequality implies that for most $j$ with $x_j = b$, $\pi_b$ and~$\sigma_j$ have close Pauli coordinates.  Finally, average back over the transcripts to get that for most $j$ with $x_j = b$, $\pi_b$ is close to $\rho_j = \sum_{\hA{11(j-1)}} \!\!\! \Pr[\hA{11(j-1)}] \sigma_j(\hA{11(j-1)})$.  This is the claim.  

\thmref{t:statetomographyroughly} says that Bob's measurement usually prepares the correct state on Alice's side---but it says nothing about the distribution of his measurement results.  Since Bob's half of the shared EPR states is maximally mixed, his measurement outcomes are in fact distributed nearly uniformly on most subsystems.  Thus the effect of Bob's actual super-operator is close to that of the ideal super-operator, if we trace out everything except for a random subset of subsystems on Alice's side.

\subsection{Process tomography}

Computation by teleportation uses adaptively chosen two-qubit Bell measurements on prepared resource states.  The state tomography protocol gives Eve a way of ensuring that Alice's initial $m$-qubit state consists of the desired resource states.  The Bell states $\big\{ (P \otimes I) \ket \varphi : P \in \{I, \sigma_x, \sigma_y, \sigma_z\} \big\}$, eigenstates of the Bell measurement, are themselves $X\!Z$-determined states.  A symmetrical protocol, with the roles of Alice and Bob switched, could thus be used to prepare these states in Bob's initial $m$-qubit state, before his play begins.  However, for verified quantum computation we need a stronger conclusion, that \emph{after} Bob's measurements collapse Alice's initial state into copies of the resource states of Eq.~\eqnref{e:resourcestatebasis}, Alice's operations indeed implement Bell measurements.  This is a stronger conclusion because, as explained earlier, it is possible that dishonest operations on Alice's side---far from Bell measurements---could still collapse Bob's initial state to many copies of Bell states.  

The process tomography protocol, to ensure Alice's honesty, begins in the same way as state tomography, with $(K-1) m$ CHSH games refereed between both players.  The multi-game rigidity theorem again allows us to assume that the devices then share $m$ EPR states and would play the $K$th block of games honestly, according to the ideal CHSH game strategy.  However, Eve only continues playing CHSH games with Bob.  To Alice, Eve repeatedly sends random pairs of indices, drawn uniformly without replacement from $\{1, \ldots, m\}$, requesting the results of Bell measurements on the specified qubits.  Eve does not try to draw any statistical conclusions from the reported results, but instead rejects if Alice's reported result for any pair of qubits is inconsistent with Bob's reported outcomes.  For example, if Alice reports measuring the Bell state $\ket \varphi$ and Bob is asked to measure the same two qubits both in the $\sigma_z$ basis, then the parity of Bob's measurement results should be~$0$.  This is because $\sigma_z \otimes \sigma_z$ is a stabilizer for $\ket \varphi$: $(\sigma_z \otimes \sigma_z) \ket \varphi = \ket \varphi$.  The four Bell states are all stabilized by $\pm \sigma_z \otimes \sigma_z$ and $\pm \sigma_x \otimes \sigma_x$, with four different sign choices.  If Bob is asked to measure one of the two qubits in the $\sigma_x$ basis and one in the $\sigma_z$ basis, then his results for that pair do not give a check on Alice.  

From Alice's perspective, the interaction is indistinguishable from the questions she would be asked in an actual computation; she is entirely unaware that Bob is playing CHSH games instead of preparing resource states.  

We show: 

\begin{theorem}
In the process tomography protocol, if after the first $(K-1) m$ CHSH games, Bob plays honestly on his halves of $m$ shared EPR states, then: 
\begin{description}
\item[Completeness:]
If Alice plays honestly, Eve accepts with probability one.  
\item[Soundness:]
If Eve accepts with high probability, then the result of Alice's super-operator applied to the initial state is close to that of applying the ideal Bell measurements super-operator to the initial state.  
\end{description}
\end{theorem}

Again, the completeness statement is trivial.  The soundness statement is not difficult.  Fix a permutation of the qubits for which Eve accepts with high probability.  We prove soundness for the protocol in which Alice is given the full permutation at the beginning instead of only two indices at a time; this can only give her more opportunities to cheat.  

\def\RB {R^B}
\def\RBdeca #1#2{#1{R}^B_{#2}}
\def\RBdec #1{#1{R}^B}
\def\PAdecjh #1#2#3{#1{P}^A_{#2}({#3})}
\def\PBdecjh #1#2#3{#1{P}^B_{#2}({#3})}
\def\PBdech #1#2{#1{P}^B({#2})}
\def\AmeasBdecj #1#2{#1{\F}^A_{#2}}
\def\AmeasBBdecj #1#2{#1{\F}^B_{#2}}	% this is the same as above, but I have changed the superscript to B for the supplementary information (in the main text, we use A still)
\def\GAj #1{\G^A_{#1}}
\def\GBj #1{\G^B_{#1}}
\def\GAdecj #1#2{#1{\G}^A_{#2}}
\def\GBdecj #1#2{#1{\G}^B_{#2}}

Let $\rhodecone{\hat}$ be the initial state, consisting of $m$ EPR states.  Without loss of generality, Alice's strategy consists of measuring a complete set of $2^m$ orthogonal projections, and returning the outcome.  For $j = 1, \ldots, m/2$, let $\GAj{j}$ be Alice's super-operator that implements the $j$th alleged Bell measurement, by the appropriate marginal projective measurement.  For $j \leq k$, let $\GAj{j,k} = \GAj{k} \cdots \GAj{j+1} \GAj{j}$.  Alice's full strategy is implemented by $\GAj{1,m/2}$.  Let $\GAdecj{\hat}{j}$ be the ideal super-operator that actually carries out a Bell measurement on the $j$th specified pair of qubits, and $\GAdecj{\hat}{j,k} = \GAdecj{\hat}{k} \cdots \GAdecj{\hat}{j}$.  Our goal is to show that in trace distance 
\begin{equation}
\GAj{1,m/2}(\rhodecone{\hat}) \approx \GAdecj{\hat}{1,m/2}(\rhodecone{\hat})
 \enspace .
\end{equation}

Since Eve accepts all the tests with high probability, it must be that for every $j$, Eve's $j$th Bell measurement test passes with high probability.  By the Gentle Measurement Lemma, this implies that $\GAj{j}(\rhodecone{\hat}) \approx \GAdecj{\hat}{j}(\rhodecone{\hat})$.  For $j = 1$, this gives the first step: 
\begin{equation*}
\GAj{1,m/2}(\rhodecone{\hat}) \approx \GAj{2,m/2} \GAdecj{\hat}{1}(\rhodecone{\hat})
 \enspace .
\end{equation*}
We cannot immediately apply $\GAj{2}(\rhodecone{\hat}) \approx \GAdecj{\hat}{2}(\rhodecone{\hat})$ to continue, because although the $\GAj{j}$ super-operators commute with each other, $\GAj{2}$ might not commute with $\GAdecj{\hat}{1}$.  

An easy trick gets around the problem.  Define $\AmeasBBdecj{\hat}{j}$ to implement a Bell measurement on the $j$th specified pair of qubits on Bob's side.  Let $\AmeasBBdecj{\hat}{j,k} = \AmeasBBdecj{\hat}{k} \cdots \AmeasBBdecj{\hat}{j}$.  Then $\GAdecj{\hat}{j}(\rhodecone{\hat}) = \AmeasBBdecj{\hat}{j}(\rhodecone{\hat})$.  Super-operators acting on Bob's Hilbert space automatically commute with those acting on Alice's space, so we can continue the above derivation: 
\begin{equation*}\begin{split}
\GAj{1,m/2}(\rhodecone{\hat}) 
&\approx \GAj{2,m/2} \GAdecj{\hat}{1}(\rhodecone{\hat})
= \AmeasBBdecj{\hat}{1} \GAj{2,m/2}(\rhodecone{\hat}) \\
&\approx \AmeasBBdecj{\hat}{1} \GAj{3,m/2} \GAdecj{\hat}{2}(\rhodecone{\hat}) = \AmeasBBdecj{\hat}{1,2} \GAj{3,m/2}(\rhodecone{\hat}) \approx \cdots \\
&\approx \AmeasBBdecj{\hat}{1,m/2}(\rhodecone{\hat}) = \GAdecj{\hat}{1,m/2}(\rhodecone{\hat})
 \enspace .
\end{split}\end{equation*}
The total approximation error is linear in~$m$.

\subsection{$\QMIP = \MIP^*$}

The way in which protocols for CHSH games, state and process tomography, and computation are combined to give verified quantum dynamics is sketched in the main text.  We also describe there the main remaining technical obstacle, the issue of Eve choosing her questions adaptively.  
Formally, let $\rho$ be the initial state, and let~$\B$ be the super-operator describing Eve's interactions with Bob in state tomography.  Roughly, state tomography implies that the states Bob prepares on Alice's side are correct up to a small error in trace distance, or 
\begin{equation}
\Tr_B \B(\rho) \approx \Tr_B \hat \B(\hat \rho)
 \enspace ,
\end{equation}
where $\hat \B$ is the ideal super-operator and $\hat \rho$ is an ideal initial state consisting of perfect EPR states.  Similarly, let $\A$ be the super-operator describing Eve's interactions with Alice in a process tomography protocol on Alice's operations; we have 
\begin{equation}
\A(\rho) \approx \hat \A(\rho)
 \enspace .
\end{equation}
Computation by teleportation can be implemented either by choosing Bob's state preparation questions non-adaptively and Alice's process questions adaptively, or vice versa.  We show that these are exactly equivalent regardless of the devices' strategies, i.e., 
\begin{equation}
\Aad \B = \Bad \A
 \enspace ,
\end{equation}
where $\Aad$ and $\Bad$ are the same as~$\A$ and~$\B$, respectively, except with Eve choosing her questions adaptively based on the previous messages.  
Combining these steps, we therefore obtain 
\begin{equation*}\begin{split}
\Tr_B \Bad \A (\rho) 
&\approx \Tr_B \Bad \hat \A (\rho) \\
&= \Aadhat \Tr_B \B (\rho) \\
&\approx \Aadhat \Tr_B \Badhat (\hat \rho)
 \enspace ,
\end{split}\end{equation*}
and thus the actual computation by teleportation protocol leaves on Alice's side nearly the ideal output.  

In this supplement, we would like to highlight two points of the proof that $\QMIP = \MIP^*$: the conversion into a three-round protocol, and the addition of two new provers instead of reusing provers.  

$\QMIP$ is the class of languages decidable by a polynomial-time quantum verifier exchanging polynomially many quantum messages with a polynomial number of quantum provers, who have unbounded computational power and share entanglement but cannot communicate among themselves.  Kempe et al.\ have shown that any QMIP protocol can be converted into a three-turn protocol in which the provers send a quantum message to the verifier, the verifier broadcasts the result of a random coin flip, the provers each send a second quantum message, and then the verifier applies an efficient measurement to decide whether to accept or reject~\cite{KempeKobayashiMatsumotoVidick07qmip}.  Beginning with this protocol transformation, our proof adds two new provers, Alice and Bob.  The classical verifier, Eve, teleports both rounds of messages from the original $k$ provers to Alice, and then directs Alice and Bob to run the original verifier's quantum circuit.  

A natural question is whether it is necessary to add two new provers or if two of the provers already present can be used for implementing verified quantum dynamics.  We conjecture that adding new provers is not necessary.  Broadbent et al., for example, have suggested that a $k$-prover QMIP protocol can be converted to a $k$-prover MIP$^*$ protocol deciding the same language, and the scheme that they present indeed reuses the first two provers to simulate the verifier's quantum computations~\cite{BroadbentFitzsimonsKashefi10qmip}.  However, the analysis of this scheme does not consider all ways in which provers can play dishonestly.  

In fact, any scheme with a structure along the lines presented either by us or by Broadbent et al.\ will be unsound---if it does not first convert to a three-turn protocol or use more sophisticated tricks.  Here is a general counterexample.  Begin with an arbitrary QMIP protocol $\cP$, deciding a language~$L$.  Modify the protocol by adding one new round at the end.  In this last round, the provers can each send classical messages to the verifier, which the verifier simply broadcasts back to all of the provers.  The effect of this final interaction is to allow the provers to communicate with each other.  (Since they share entanglement, they can use quantum teleportation to communicate quantum information, if desired.)  The modified protocol $\cP'$ decides the same language~$L$, with the exact same completeness and soundness parameters, though; communicating in the last step does not help the provers cheat.  

Convert the modified protocol, somehow, into an MIP$^*$ protocol $\cP''$ that uses the first two provers to simulate the verifier's quantum computations.  In particular, they simulate the final acceptance predicate, by some procedure that has traps to detect cheating---in our scheme, CHSH games or state or process tomography.  The problem is that the last round of messages in $\cP'$ is useless, but in $\cP''$ the intercommunication allows the provers to reveal to each other the traps set by the verifier.  This allows them to avoid the traps and cheat freely.  $\cP''$ is unsound.  

Converting to a three-turn protocol at the start deflects this general attack, as does using two fresh provers to run the verifier's quantum computation.  Speculatively, another way to avoid the attack might be to refresh the verifier's secrets---resetting any traps and re-hiding any quantum information---before revealing any message to a prover.  Any messages can carry information that affect the security of the converted MIP$^*$ protocol differently from the original QMIP protocol.

\ifx\compilefullpaper\undefined  
\bibliographystyle{naturemag}
\bibliography{q}

\end{document}
\fi

\bibliographystyle{naturemag}
\small
\bibliography{q}

\noindent {\bf Acknowledgements}
We thank Edgar Bering, Anne Broadbent, Andr{\' e} Chailloux, Matthias Christandl, Roger Colbeck, Tsuyoshi Ito, Robert K{\"o}nig, Matthew McKague, Vidya Madhavan, Renato Renner, Shivaji Sondhi and Thomas Vidick for helpful conversations.  Part of the work conducted while F.U.~was at UC Berkeley, and B.R.~at the Institute for Quantum Computing, University of Waterloo.  B.R.~acknowledges support from NSERC, ARO-DTO and Mitacs.  U.V.~acknowledges support from NSF grant CCF-0905626 and Templeton grant~21674.

\end{document}